\def\SUPPRESSACKNOWLEDGMENT{1}
\documentclass[onecolumn,draftclsnofoot]{IEEEtran}

\usepackage{cite}
\usepackage{color}
\usepackage{graphicx}
\usepackage{placeins}
\usepackage{amsmath,amssymb,amsthm,mathrsfs,mathtools}
\usepackage{tikz}
\usepackage{url}
\usetikzlibrary{positioning,calc,arrows.meta,decorations.pathreplacing}

\let\oldbrace\{
\def\{{\oldbrace\kern0.5pt}

\def\Cov{\mathop{\rm Cov}\nolimits}%
\def\ker{\mathop{\rm ker}\nolimits}%
\newcommand{\EE}{\mathsf E}
\newcommand{\PP}{\mathsf P}

\newcommand{\ctx}{\mathrm{ctx}}

\newcommand{\suppset}{\operatorname{supp}}
\newcommand{\Areg}{\mathcal A_{\rm reg}}

\newtheorem{theorem}{Theorem}
\newtheorem{lemma}{Lemma}
\newtheorem{proposition}{Proposition}
\newtheorem{corollary}{Corollary}
\theoremstyle{definition}
\newtheorem{definition}{Definition}

\newtheorem{remark}{Remark}

\newcommand{\draftEightBegin}{}
\newcommand{\draftEightEnd}{}

\begin{document}

\title{On the Information-Theoretic Limits of Latent-Space Watermarking Through Pretrained Generators}
\ifdefined\TITVERSION
\author{Jinwan Jeon,~\IEEEmembership{Student Member,~IEEE}, Minju Lee,~\IEEEmembership{Student Member,~IEEE}, and Sung Hoon Lim,~\IEEEmembership{Senior Member,~IEEE}%
\thanks{J. Jeon, M. Lee, and S. H. Lim are with Hallym University, CHUNCHEON-SI, GANGWON-DO, South Korea (e-mail: jeonjw@hallym.ac.kr; mjlee@hallym.ac.kr; shlim@hallym.ac.kr). Sung Hoon Lim is the corresponding author.}}
\else
\author{Jinwan Jeon, Minju Lee, and Sung Hoon Lim%
\thanks{J. Jeon, M. Lee, and S. H. Lim are with Hallym University, CHUNCHEON-SI, GANGWON-DO, South Korea (e-mail: jeonjw@hallym.ac.kr; mjlee@hallym.ac.kr; shlim@hallym.ac.kr). Sung Hoon Lim is the corresponding author.}}
\fi
\maketitle

\begin{abstract}
\draftEightBegin
We study latent-space watermarking through a pretrained generator using a prescribed latent-to-output stochastic mapping, called the renderer. A watermark encoder selects the latent input using a message and secret key. For every message and semantic context, the released output must have exactly the desired conditional output distribution. For finite alphabets, we derive rate--key inner and outer bounds and characterize the coding and coordination requirements for realizing watermark communication through the prescribed latent interface. When the target output distribution of the generator uniquely determines the corresponding latent input distribution through the renderer, a strengthened converse yields the capacity region; the same region governs explicit preservation of the pretrained latent distribution. We extend the analysis to general jointly Gaussian models and identify a sufficient statistic of the latent that captures both the watermark-bearing information available at the generated output and the latent coordination required to preserve its target distribution. For the vector Gaussian model, we further characterize the optimal allocation of the secret-key resource across the resulting modes. Finally, we turn to an emerging robustness threat that is particularly natural in generative watermarking: an adversary can regenerate the released sample to obtain a fresh realization of the same underlying content while attenuating or destroying the embedded watermark. We incorporate this robustness axis into our framework and characterize the one-pass compound capacity of the scalar Gaussian model when the semantic context is known to the encoder but hidden from the detector, while the regeneration attack may depend on that context. Extending the analysis to multiple rounds of repeated canonical regeneration, we characterize the resulting watermark-capacity decay.
\draftEightEnd
\end{abstract}


\section{Introduction}
\label{sec:intro}

Generative models increasingly produce images, audio, video, and text whose provenance is difficult to establish from the content alone. Watermarking addresses this problem by embedding information that can later be detected or recovered. Alongside classical information-hiding methods and learned post-hoc image watermarks~\cite{Cox1997SpreadSpectrum,MoulinOSullivan2003,Zhu2018HiDDeN,Tancik2020StegaStamp}, recent methods integrate watermarking into language-model sampling and diffusion-based generation~\cite{Kirchenbauer2023LLMWatermark,Kuditipudi2024DistortionFree,ChristGunnZamir2024,Ho2020DDPM,Song2021SDE,Rombach2022LDM,Fernandez2023StableSignature,Rezaei2024LaWa,Gunn2025Undetectable}.

A particularly attractive architecture is \emph{latent-space watermarking}: information is inserted into the latent randomness driving a pretrained generator, while the generator remains unchanged. Tree-Ring modifies the initial diffusion noise~\cite{Wen2023TreeRing}; Gaussian Shading uses distribution-preserving sampling of a Gaussian latent~\cite{Yang2024GaussianShading}; and related work studies identification and extraction from such watermarks~\cite{Ci2024RingID,Yang2026FARI}. Working through the generator's native latent interface avoids redesigning the generation pipeline and can preserve its ordinary sampling behavior. However, selecting distinguishable latent inputs is not enough. The watermark is embedded in latent space before generation; it must then survive a stochastic mapping that the watermark designer does not control so that a legitimate party possessing the key can recover it from the generated output.

In recent work, Guo et al.~\cite{Guo2026Covert} studied an information-theoretic approach to generative watermarking by combining Gel'fand--Pinsker coding with channel synthesis. Building on this perspective, we specialize the setting to ask what changes when the watermark must be embedded through the \emph{fixed latent interface} of a pretrained generator.

The central question is therefore
\begin{quote}
\emph{What are the fundamental limits of reliably embedding watermark information through the fixed latent interface of a pretrained generator while preserving its ordinary conditional generation distribution?}
\end{quote}
We model that interface by a fixed renderer $q_{\rm r}(x|u,s)$, where $S$ is semantic context, $U$ the physical latent, and $X$ the released sample. The encoder controls how the latent interface is driven, not how the renderer transforms latent variation into public output. This restriction therefore shapes both the design of reliable watermark communication and the coordination needed to preserve the target generation distribution.

Preserving generation quality supplies the second part of the question. We require the released sample to have exactly the desired conditional distribution for every watermark message, after averaging the secret key. Exact statistical invisibility is classical in steganography and information hiding~\cite{Cachin1998Steganography,HopperLangfordAhn2002,WangMoulin2008Perfect,BarronChenWornell2003}; coupling-based and distortion-free constructions also appear in modern generative watermarking~\cite{DeWitt2023MEC,Tsur2025Couplings,Kuditipudi2024DistortionFree}. Our criterion combines finite-block distributional equality with message-wise, context-conditioned preservation. It is stronger than the asymptotically vanishing message-wise total-variation criterion of~\cite{Guo2026Covert}; the central question here is how exact preservation and reliable communication interact when generation must pass through a fixed renderer.

The resulting key requirement has an operational interpretation. The encoder must coordinate a watermark representation with the physical latent, whereas the detector sees only what the renderer releases. A chosen representation can therefore contain latent dependence that must be coordinated before the stochastic generator but remains unresolved from the public sample. The question is not whether every unused latent coordinate consumes key; it is which dependence the watermark representation needs, and which part of it the renderer makes useful for decoding the watermark.

Generative watermarking also introduces a robustness threat that is particularly natural to the generative setting. Rather than perturbing the released sample directly, an adversary can regenerate it to obtain a fresh realization of the same underlying content while attenuating or removing the embedded watermark. In the latent-space setting, such an attack can map the released sample back to a latent representation and invoke the generator again. Empirical work shows that generative regeneration can degrade or remove invisible watermarks~\cite{An2024WAVES,Zhao2024Removable}. Because our framework already treats watermark communication through a prescribed generator, it provides a natural basis for studying this regeneration mechanism as a robustness problem. We therefore ask how much watermark information can survive repeated use of the same prescribed renderer when each regenerated sample must retain the desired distribution and satisfy a per-pass distortion constraint.

\noindent\textit{Contributions.}\ The paper develops four primary answers to these questions.
\begin{enumerate}
\item \emph{Fixed-interface rate--key limits.} We derive finite-alphabet inner and outer bounds and explain their resource structure through context adaptation, detector packing, and latent coordination. For the general finite-alphabet setting, a random-projection example shows that the outer bound can be strictly loose, while the inner bound is capacity-achieving for that example.
\item \emph{Capacity through latent-distribution identification.} We specialize our bounds to the unique-preimage condition in which the fixed renderer and the target output distribution uniquely determine the latent input distribution. That is, only one latent distribution is compatible with the pretrained generation model. For this specialization, we obtain the capacity region and an equivalent interpretation of the secret key as a latent-coordination budget. We further show that the same capacity region applies when preservation of the pretrained latent distribution is imposed explicitly, in which case the unique-preimage condition is unnecessary.
\item \emph{Gaussian renderer geometry.} We extend our results to the general jointly Gaussian vector setting. We show that the physical latent can be reduced losslessly to its renderer-visible effect, revealing a sufficient statistic of the latent for watermark communication and generation coordination. We characterize the finite-key capacity, show that hiding the context at the detector causes no loss, and derive the optimal secret-key allocation across renderer-visible modes. A scalar case study explains the within-mode coordination tradeoff before we turn to the vector allocation problem.
\draftEightBegin
\item \emph{Regeneration robustness.} We also extend the framework to a same-generator regeneration attack, an emerging threat in which a released sample is regenerated through the generative model in an attempt to weaken or remove the watermark. Focusing on this attack family, we show that a fixed attack changes the effective renderer and therefore allows the base coding theory to be reused. For the scalar Gaussian model with random context known at the encoder and hidden from the detector, we characterize the minimum distortion of a genuine rerender and the compound capacity achieved by a single code and detector over the full admissible one-pass attack family. We then extend the analysis to repeated canonical Gaussian regeneration and characterize the resulting capacity decay and the number of regeneration rounds compatible with a positive target rate. The unrestricted multiround compound minimax problem remains outside the present scope.
\draftEightEnd
\end{enumerate}

Section~\ref{sec:model} defines the fixed interface and preservation criterion. Sections~\ref{sec:finite} and~\ref{sec:up} develop the general bounds, their interpretation, and the capacity region under latent-distribution identification. Section~\ref{sec:gaussian} gives the Gaussian theory, Section~\ref{sec:robust} studies regeneration, and Section~\ref{sec:discussion} draws together the implications. The proofs are in the appendices.

\noindent\textit{Notation.}\ We closely follow the notation in~\cite{ElGamalKim2011}. Random variables are denoted by uppercase letters and their realizations by the corresponding lowercase letters; sets and alphabets use calligraphic letters. We use $p$ for a generic probability mass function or density, with arguments and conditioning identifying its joint, marginal, or conditional form. Distinguished distributions are assigned separate symbols when their roles require it. Probabilities of events and expectations are denoted by $\PP$ and $\EE$, respectively; measure notation is used when a density need not exist. A length-$n$ sequence is $X^n=(X_1,\ldots,X_n)$, and $X\perp Y$ denotes independence. For random vectors, $\Sigma_X$ and $\Sigma_{X|Y}$ denote covariance and conditional covariance matrices, respectively, and $A\succ0$ means that the symmetric matrix $A$ is positive definite. All logarithms are to base two and information quantities are measured in bits. For a finite alphabet $\mathcal A$, $\Delta(\mathcal A)$ denotes its probability simplex.

\section{Problem Formulation}
\label{sec:model}

We specify the pretrained generator, then the watermarking code and its preservation requirement.

\subsection{Pretrained Generator and Latent Interface}
\label{sec:model-generator}

Let $(S,X)\sim q_0(s,x)$ be a memoryless source pair on $\mathcal S\times\mathcal X$, where $S$ represents the semantic/context state and $X$ the corresponding observable/generated sample. We call $q_0$ the \emph{desired distribution}, using the same symbol for its joint, marginal, and conditional forms. The pair specifies the desired public statistics: in the watermarking problem, the encoder observes $S^n$ and generates $X^n$ through the pretrained system, rather than observing $X^n$ beforehand.

For one invocation, the pretrained system realizes $q_0(x|s)$ through a physical latent $U\in\mathcal U$. Its conditional latent prior is described by $q_{\rm lat}(u|s)$, which we call the \emph{pretrained latent law}. Given $(u,s)$, the system generates $X$ through the fixed stochastic mapping $q_{\rm r}(x|u,s)$, called the \emph{renderer}. In channel terms, this is a fixed state-dependent channel from the latent $U$ to the public output $X$. Thus ordinary generation follows
\begin{equation}
U\sim q_{\rm lat}(\cdot|S),
\quad
X\sim q_{\rm r}(\cdot|U,S),
\label{eq:ordinary-generation}
\end{equation}
and the \emph{pretrained-system joint law} is
\begin{equation}
q_{\rm pre}(s,u,x)
=q_0(s)q_{\rm lat}(u|s)q_{\rm r}(x|u,s).
\label{eq:pretrained-joint}
\end{equation}
Its $(S,X)$ marginal is the desired distribution $q_0(s,x)$; equivalently, for every relevant $s$,
\begin{equation}
q_0(x|s)
=\sum_u q_{\rm lat}(u|s)q_{\rm r}(x|u,s).
\label{eq:q0}
\end{equation}
Sums have the corresponding integral/kernel interpretation on continuous alphabets.

The renderer is inherited from the pretrained system and transferred unchanged into the watermarking problem. In watermark design, however, the latent-input distribution $p(u|s)$ is designable. Choosing $p(u|s)=q_{\rm lat}(u|s)$ preserves the pretrained latent law, but this equality is not required in the base problem. Once $q_0$ and $q_{\rm r}$ are fixed, the latent-input distribution may change as long as the prescribed output distribution is preserved through the fixed renderer.

\subsection{Watermark Code and Performance Criteria}

At blocklength $n$, the \emph{watermark latent encoder} is an arbitrary conditional distribution
\begin{equation}
p(u^n|s^n,m,k),
\label{eq:watermark-latent-encoder}
\end{equation}
where $M$ is the watermark message and $K$ the secret key, with
\begin{equation}
M\sim\mathrm{Unif}[1:2^{nR}],
\quad
K\sim\mathrm{Unif}[1:2^{nR_k}].
\label{eq:message-key-laws}
\end{equation}
The message, key, and context sequence are mutually independent, and $q_0(s^n)=\prod_{i=1}^n q_0(s_i)$. The encoder may couple its latent inputs across the block. Its stochasticity subsumes encoder-private randomization, with deterministic encoding as a special case.

The renderer operates memorylessly:
\begin{equation}
q_{\rm r}(x^n|u^n,s^n)
=\prod_{i=1}^n q_{\rm r}(x_i|u_i,s_i).
\label{eq:renderer-product}
\end{equation}
The detector observes the released sample $X^n$ and the secret key $K$, but not $S^n$.

\begin{definition}[Watermarking code]
\label{def:watermark-code}
An $(n,R,R_k;q_0,q_{\rm r})$ \emph{watermarking code} consists of a watermark latent encoder~\eqref{eq:watermark-latent-encoder}. 
The watermarked latent $U^n$ with context $S^n$ is processed through the fixed renderer~\eqref{eq:renderer-product} which outputs $X^n$. 
A detector applies the mapping
\begin{equation}
g_n:\mathcal X^n\times[1:2^{nR_k}]\to[1:2^{nR}].
\end{equation}
The detector produces $\hat M=g_n(X^n,K)$ and does not observe the semantic context $S^n$.
\end{definition}

With the uniform message and key distributions denoted by $p(m)$ and $p(k)$, the full joint distribution induced by a code is
\begin{align}
&p(s^n,m,k,u^n,x^n)\nonumber\\
&\quad=q_0(s^n)p(m)p(k)
 p(u^n|s^n,m,k)q_{\rm r}(x^n|u^n,s^n).
\label{eq:code-induced-joint}
\end{align}
In particular, averaging over the key and the encoder and renderer randomness gives
\begin{align}
&p(x^n|s^n,m)\nonumber\\
&\quad=\sum_{k,u^n}p(k)p(u^n|s^n,m,k)q_{\rm r}(x^n|u^n,s^n),
\label{eq:induced-output-distribution}
\end{align}
which we call the \emph{induced output distribution}.

A rate pair $(R,R_k)$ is \emph{achievable} if there exists a sequence of blocklength-$n$ watermarking codes satisfying
\begin{equation}
\PP\{\hat M\ne M\}\to0
\label{eq:reliability}
\end{equation}
and, at every blocklength and for every message $m$,
\begin{equation}
p(x^n|s^n,m)
=\prod_{i=1}^n q_0(x_i|s_i),
\label{eq:exact-output}
\end{equation}
that is, the code induced output distribution matches exactly with the target law $q_0(x^n|s^n)$ for every fixed message. 
The equality holds for almost every $s^n$ under $q_0(s^n)$; on finite alphabets, it holds for every $s^n$ with $q_0(s^n)>0$. Distributional equalities are interpreted as equalities of measures, or of densities where these exist. The conditioning in~\eqref{eq:exact-output} does not include the key: the induced output distribution is the key-averaged distribution in~\eqref{eq:induced-output-distribution}.

The \emph{capacity region} $\mathcal C$ is the closure of the set of achievable rate pairs. Concatenating exact codes preserves~\eqref{eq:exact-output}, so coded time sharing makes $\mathcal C$ convex; unused key bits imply that it is upward closed in $R_k$. For a fixed key-rate budget $R_k$, define its upper envelope by
\[
C(R_k)
\triangleq
\sup\{R:(R,R_k)\in\mathcal C\}.
\]
When the detector additionally observes $S^n$ (benchmark case), we use $\mathcal C_{\ctx}$ and $C_{\ctx}(R_k)$ for the corresponding capacity region and upper envelope, with the encoder, renderer, and exact preservation criterion unchanged.

We call~\eqref{eq:exact-output} \emph{exact output preservation}. It serves two related roles. It is a message-wise \emph{watermark invisibility} requirement, since the message produces no distributional signature for an observer without the key. It also preserves the pretrained generator's desired conditional output statistics for every relevant context. The desired distribution $q_0$ is fixed by the problem, whereas $p(x^n|s^n,m)$ is induced by the chosen code; the constraint requires these public statistics to agree exactly. It does not require the code to preserve $q_{\rm lat}$. Wang and Moulin's perfect-security criterion requires equality of the covertext and stegotext block marginals~\cite{WangMoulin2008Perfect}. Here equality is imposed separately for each message and context; it is also finite-block exact, rather than asymptotically vanishing message-wise total variation as in~\cite{Guo2026Covert}. The distinction of interest here is not exactness alone, but exactness when watermark communication must be realized through the prescribed latent renderer.

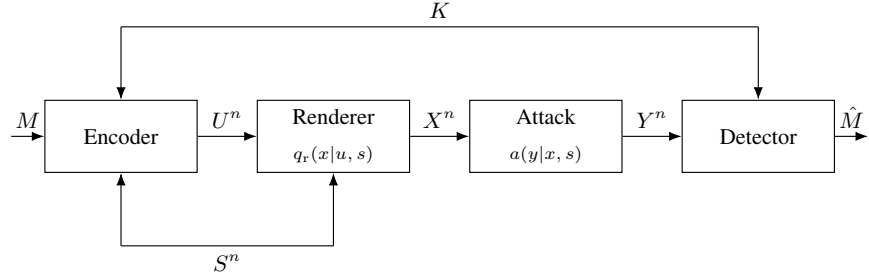
\begin{figure}[t]
\centering
\resizebox{0.70\linewidth}{!}{
\begin{tikzpicture}[
  >=Latex,
  line width=0.45pt,
  block/.style={draw, rectangle, minimum width=2.15cm, minimum height=1.02cm, inner sep=2pt, align=center},
  flow/.style={->, line width=0.45pt},
  branch/.style={->, line width=0.45pt},
  every node/.style={font=\small}
]
  \node[block] (enc) at (0,0) {Encoder};
  \node[block] (ren) at (3.00,0) {Renderer\\[-1pt]{\scriptsize $q_{\rm r}(x|u,s)$}};
  \node[block] (att) at (6.00,0) {Attack\\[-1pt]{\scriptsize $a(y|x,s)$}};
  \node[block] (dec) at (9.00,0) {Detector};

  \draw[flow] (-1.55,0) -- node[above, midway] {$M$} (enc.west);
  \draw[flow] (enc.east) -- node[above, midway] {$U^n$} (ren.west);
  \draw[flow] (ren.east) -- node[above, midway] {$X^n$} (att.west);
  \draw[flow] (att.east) -- node[above, midway] {$Y^n$} (dec.west);
  \draw[flow] (dec.east) -- node[above, midway] {$\hat M$} (10.55,0);

  \coordinate (ktopenc) at (0,1.55);
  \coordinate (ktopdec) at (9.00,1.55);
  \draw[line width=0.45pt] (ktopenc) -- (ktopdec);
  \node[above] at (4.50,1.55) {$K$};
  \draw[branch] (ktopenc) -- (enc.north);
  \draw[branch] (ktopdec) -- (dec.north);

  \coordinate (sbotenc) at (0,-1.55);
  \coordinate (sbotren) at (3.00,-1.55);
  \draw[line width=0.45pt] (sbotenc) -- (sbotren);
  \node[below] at (1.50,-1.55) {$S^n$};
  \draw[branch] (sbotenc) -- (enc.south);
  \draw[branch] (sbotren) -- (ren.south);
\end{tikzpicture}}
\caption{Latent-space watermarking through a fixed pretrained renderer and the same-generator attack extension. In the no-attack baseline the attack block is omitted, equivalently $Y^n=X^n$. The possible context dependence of the attack is suppressed in the diagram and written explicitly in Sec.~\ref{sec:model-attack}.}
\label{fig:model}
\end{figure}

As a final note on the formulation, the key restriction is that the watermark is embedded through a renderer inherited from a pretrained generative model and held fixed during watermark design. If the renderer were fully redesignable, the encoder could realize the desired output mechanism directly, and the latent structure would no longer constitute an operational constraint. In this fully designable limit, the model reduces to the output-watermarking formulation studied by Guo et al.~\cite{Guo2026Covert}. Thus, the distinction considered here is precisely the restriction imposed by watermarking through a fixed pretrained renderer. The regeneration-attack model and its associated robustness criterion are introduced separately in Section~\ref{sec:robust}.

\section{General Rate--Key Limits for Latent-Space Watermarking}
\label{sec:finite}

In this section, we first state the general finite-alphabet bounds, then expose their coding mechanism and the additional information needed for a capacity converse.

Let $V$ denote an auxiliary random variable, distinct from the physical latent $U$, with single-letter conditional distribution
\begin{equation}
p(v,u|s)=p(v|s)p(u|s,v).
\label{eq:design-factorization}
\end{equation}
The latent-completion kernel $p(u|s,v)$ maps the auxiliary $V$, together with the context $S$, into the physical latent alphabet $\mathcal U$ inherited from the pretrained renderer. Together with $p(v|s)$, it determines the joint conditional distribution~\eqref{eq:design-factorization}. This is a single-letter design object, not a claim that the operational watermark encoder acts symbol by symbol.

The resulting full joint distribution is restricted to the class
\begin{equation}
\mathscr P(q_0;q_{\rm r})
=\left\{\begin{array}{l}
p(s,v,u,x)=q_0(s)p(v,u|s)q_{\rm r}(x|u,s):\\
\hfill p(s,x)=q_0(s,x)
\end{array}\right\}.
\label{eq:test-class}
\end{equation}
Thus $p(u|s)=\sum_v p(v,u|s)$ must satisfy $\sum_u p(u|s)q_{\rm r}(x|u,s)=q_0(x|s)$, but need not equal $q_{\rm lat}(u|s)$. Unless stated otherwise, the information quantities below are evaluated under the selected distribution $p\in\mathscr P(q_0;q_{\rm r})$.

\subsection{General Finite-Alphabet Bounds}

\begin{theorem}[General finite-alphabet rate--key bounds]
\label{thm:finite}
Let $S\sim q_0(s)$ be the semantic source, $q_{\rm r}(x|u,s)$ be the fixed renderer, and $q_0(x|s)$ be the desired conditional output distribution. Every rate pair in the convex hull of all rate pairs satisfying
\begin{align}
R
&< I(V;X)-I(V;S),
\label{eq:inner-rate}\\
R_k
&> R+I(V;S,U|X)
\label{eq:inner-key}
\end{align}
for some $p\in\mathscr P(q_0;q_{\rm r})$ is achievable. Conversely, every achievable rate pair $(R,R_k)$ must satisfy
\begin{align}
R
&\le I(V;X)-I(V;S),
\label{eq:outer-rate}\\
R_k
&\ge R+I(V;S|X)
\label{eq:outer-key}
\end{align}
for some $p\in\mathscr P(q_0;q_{\rm r})$. The auxiliary alphabet may be restricted to $|\mathcal V|\le |\mathcal S||\mathcal U|+1$.
\end{theorem}
The familiar Gel'fand--Pinsker form of the payload expression does not remove the renderer restriction: the optimization is over the constrained class~\eqref{eq:test-class}, not over freely chosen public-output channels.

Achievability combines a Gel'fand--Pinsker likelihood-encoding/soft-covering construction, which coordinates the physical latent before the fixed renderer, with an encoder-local maximal-coupling step that enforces exact output preservation. The converse uses Fano's inequality and the exact product distribution of $(S^n,X^n)$ to obtain the single-letter outer bound. Complete proofs are given in Apps.~\ref{app:finite-ach} and~\ref{app:finite-conv}.

The outer region is convex because time sharing between admissible distributions preserves the common $(S,X)$ marginal $q_0$, so both outer-bound functionals average linearly. The single-letter inner region need not be convex, and its convex hull is achievable by coded time sharing.

\begin{remark}[Loose converse]
\label{rem:finite-converse-tightness}
The outer bound in Thm.~\ref{thm:finite} is not tight in general; Prop.~\ref{prop:strict-converse-example} in App.~\ref{app:finite-conv} gives an explicit strict example.
\end{remark}

\subsection{From Coding Resources to the Key Requirement}

We next summarize the coding construction at the level needed to interpret the rate--key bounds. The discussion isolates the covering, packing, and latent-coordination requirements that determine the key resource.

For a fixed admissible $p(v,u|s)$, generate codewords
\[
V^n(m,\ell,k),\qquad
m\in[1:2^{nR}],\quad
\ell\in[1:2^{n\hat R}],\quad
k\in[1:2^{nR_k}].
\]
The encoder selects the covering index $L$ so that $V^n(M,L,K)$ covers the context sequence $S^n$. For each fixed $m$, the aggregate key and covering indices synthesize the required pair $(S^n,U^n)$. Given the key, the detector decodes $(m,\ell)$ through the renderer output. These tasks give
\begin{align}
\hat R&>I(V;S),
&&\text{covering},
\label{eq:resource-cover}\\
R+\hat R&<I(V;X),
&&\text{packing},
\label{eq:resource-pack}\\
R_k+\hat R&>I(V;S,U),
&&\text{latent synthesis}.
\label{eq:resource-synthesis}
\end{align}
These three conditions concern the same codebook and are not independent resource expenditures. Note that Fourier--Motzkin elimination of $\hat R$ from~\eqref{eq:resource-cover}--\eqref{eq:resource-synthesis} yields the inner bound in Thm.~\ref{thm:finite}.

We explain this resource allocation with the aid of Fig.~\ref{fig:key-resource-accounting}, which illustrates how the covering, packing, and latent-synthesis requirements interact in determining the message and key rates.
\begin{figure}[t]
\centering
\resizebox{0.516\linewidth}{!}{\begin{tikzpicture}[font=\small]
\def\dW{3.15}
\def\rW{2.25}
\def\tW{1.55}
\def\h{0.78}
\def\xLeft{3.85}
\pgfmathsetmacro{\xR}{\xLeft+\dW}
\pgfmathsetmacro{\xT}{\xR+\rW}
\pgfmathsetmacro{\xEnd}{\xT+\tW}

\node[anchor=east,align=right] at (3.55,3.89) {Conditional latent-\\coordination cost};
\node[anchor=east,align=right] at (3.55,0.79) {Net packing\\budget};

\draw[line width=0.9pt] (\xLeft,3.50) rectangle ++(\dW,\h);
\draw[line width=0.9pt] (\xR,3.50) rectangle ++(\rW,\h);
\draw[line width=0.9pt] (\xT,3.50) rectangle ++(\tW,\h);
\node[font=\scriptsize] at (\xLeft+0.5*\dW,3.50+0.5*\h) {$I(V;S,U|X)$};
\node at (\xR+0.5*\rW,3.50+0.5*\h) {$R$};
\node at (\xT+0.5*\tW,3.50+0.5*\h) {$t$};

\draw[line width=0.9pt] (\xR,0.40) rectangle ++(\rW,\h);
\draw[line width=0.9pt] (\xT,0.40) rectangle ++(\tW,\h);
\node at (\xR+0.5*\rW,0.40+0.5*\h) {$R$};
\node at (\xT+0.5*\tW,0.40+0.5*\h) {$t$};

\draw[densely dashed,thin] (\xR,0.22) -- (\xR,4.45);
\draw[densely dashed,thin] (\xT,0.22) -- (\xT,4.45);
\draw[densely dashed,thin] (\xEnd,0.22) -- (\xEnd,4.45);

\draw[decorate,decoration={brace,amplitude=5pt,mirror}] (\xLeft,3.30) -- (\xT,3.30)
 node[midway,below=6pt] {$R_k$};
\draw[decorate,decoration={brace,amplitude=5pt,mirror}] (\xT,3.30) -- (\xEnd,3.30)
 node[midway,below=6pt] {$t$};

\draw[decorate,decoration={brace,amplitude=5pt}] (\xR,1.34) -- (\xEnd,1.34)
 node[midway,above=6pt] {$I(V;X)-I(V;S)=R+t$};

\pgfmathsetmacro{\xMid}{0.5*(\xLeft+\xEnd)}
\node at (\xMid,-0.48) {$\hat R=I(V;S)+t$};
\end{tikzpicture}}
\caption{Boundary resource accounting for a fixed admissible $p(v,u|s)$. The net packing budget $I(V;X)-I(V;S)$ is shared by the payload rate $R$ and the excess covering rate $t$. The conditional latent-coordination cost is supplied by the same excess covering rate together with the key rate $R_k$. The aligned segments show that allocating more packing to payload leaves less available for excess covering and therefore more coordination to be supplied by the key. Equalities in the figure are boundary relations in closure.}
\label{fig:key-resource-accounting}
\end{figure}
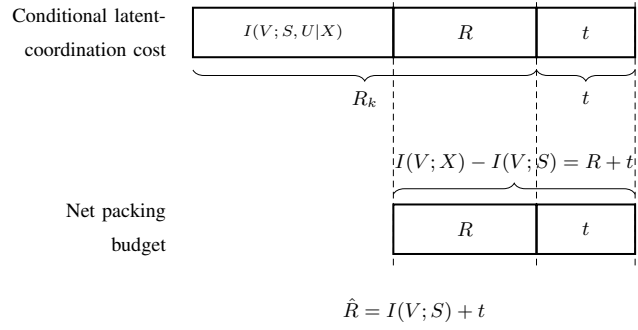
Write
\begin{equation}
\hat R=I(V;S)+t.
\label{eq:excess-multicoding}
\end{equation}
Here $t$ measures the excess auxiliary-codebook rate above that required to cover the context sequence $S^n$. The remaining constraints are
\begin{align}
R+t&< I(V;X)-I(V;S), \quad t>0,
\label{eq:net-packing}\\
R_k+t&> I(V;U|S).
\label{eq:net-coordination}
\end{align}
The right-hand side of~\eqref{eq:net-packing} is the net packing budget shared between the payload rate $R$ and the excess auxiliary-codebook rate $t$. The right-hand side of~\eqref{eq:net-coordination} is a \emph{conditional latent-coordination cost}. Allocating more of the packing budget in~\eqref{eq:net-packing} to the payload rate $R$ leaves less available for the excess covering rate $t$. The key must then supply the remaining coordination rate in~\eqref{eq:net-coordination}. Increasing $t$ does not reduce this fixed distribution's cost; it supplies more of that cost while using packing that could otherwise carry payload.

Thus $t$ and $R_k$ substitute one-for-one inside~\eqref{eq:net-coordination}, but not throughout the construction. The key is already known to the detector and does not enter its packing load. It cannot, however, replace the covering rate needed to cover the context sequence $S^n$. The payload rate $R$ plays a different role: exact preservation in~\eqref{eq:exact-output} is imposed for each fixed message, so diversity across messages cannot contribute to latent synthesis even though the payload consumes the packing budget in~\eqref{eq:net-packing}.

This gives a useful comparison with Wagner's finite-common-randomness rate--distortion--perception formulation~\cite{Wagner2022CommonRandomness}. There, the source-dependent description index also contributes reconstruction-synthesis diversity. In the present multicoding construction, its closest counterpart is the covering index $L$, not the independent watermark message $M$. Our detector must recover $L$ through the fixed renderer, so the covering index shares a packing budget with payload; the key supplies coordination without joining that decoded load.

\subsection{Latent Completion: What the Renderer Does Not Resolve}

For feasible payload under the selected distribution, the largest excess covering rate for multicoding in closure is $t=I(V;X)-I(V;S)-R$. Substituting it into~\eqref{eq:net-coordination} gives~\eqref{eq:inner-key}, since $V-(S,U)-X$. The key overhead decomposes as
\begin{equation}
I(V;S,U|X)=I(V;S|X)+I(V;U|S,X).
\label{eq:completion-decomposition}
\end{equation}
The second term,
\begin{equation}
I(V;U|S,X),
\label{eq:completion-cost}
\end{equation}
is the \emph{latent-completion cost}. It does not represent additional watermark information available to the detector; rather, it quantifies the dependence between the watermark representation $V$ and the physical latent $U$ that remains unresolved after $S$ and $X$ are known.

For example, if $U=(U_{\rm pub},U_{\rm hid})$ and $X=U_{\rm pub}$, then~\eqref{eq:completion-cost} is $I(V;U_{\rm hid}|S,X)$. Only dependence of the hidden part on the coding representation contributes. There is no charge merely for having a renderer-invisible coordinate: the encoder may choose its representation to avoid unnecessary dependence. Section~\ref{sec:gaussian} makes this point constructive by generating unused physical-latent detail privately.

At the full payload supported by a selected distribution, $R=I(V;X)-I(V;S)$ and $t=0$ in closure. The key requirement becomes
$
R_k\ge I(V;U|S)
=I(V;X|S)+I(V;U|S,X).
$
The corresponding fully designable-output construction in~\cite{Guo2026Covert} has the public-output requirement $I(V;X|S)$. The additional term here concerns realizing the representation in the physical latent before rendering. We note that this comparison is for a selected distribution and its factorization; it is not automatically the gap between two optimized regions, nor a converse obliging every possible code to use that representation.

Two endpoints sharpen the interpretation. First, exact message-wise preservation gives $M\perp X^n$, while reliable decoding from $(X^n,K)$ requires $H(M|X^n,K)\le n\epsilon_n$. Hence
$H(M)\le H(K)+n\epsilon_n,$
so $R\le R_k$ and zero key cannot support positive payload. Second, with unlimited key, the inner and outer bounds have the common payload capacity
\begin{equation}
C(\infty)=\max_{p\in\mathscr P(q_0;q_{\rm r})}
[I(V;X)-I(V;S)].
\label{eq:unlimited-general}
\end{equation}
The finite-key gap disappears, but the fixed-renderer restriction on the optimization remains.

\section{Capacity Region under a Unique-Preimage Condition}
\label{sec:up}

Exact output preservation becomes more informative when the desired output has a unique latent preimage. We first show that this one-letter property identifies the full block latent distribution, then use it to strengthen the converse.

\subsection{Unique-Preimage Condition}

For each context value $s$, define the preimage set
\begin{equation}
\mathcal P_s(q_0;q_{\rm r})
=
\left\{
\mu\in\Delta(\mathcal U):
\sum_u\mu(u)q_{\rm r}(x|u,s)=q_0(x|s),\ \forall x
\right\}.
\label{eq:preimage-set}
\end{equation}
That is, $\mathcal P_s(q_0;q_{\rm r})$ is the set of viable latent priors that, when paired with the pretrained renderer $q_{\rm r}$, induce the desired context-conditioned output distribution $q_0(x|s)$.

\begin{definition}[Unique preimage relative to the desired distribution]
\label{def:up}
The renderer satisfies the \emph{unique-preimage (UP) condition} relative to $q_0$ if, for every $s$ with $q_0(s)>0$,
\begin{equation}
\mathcal P_s(q_0;q_{\rm r})
=
\{q_{\rm lat}(\cdot|s)\}.
\label{eq:up}
\end{equation}
Since $q_{\rm lat}(\cdot|s)$ is the pretrained latent prior by construction, uniqueness of the viable latent prior implies~\eqref{eq:up}.
\end{definition}

The condition is relative to the desired distribution $q_0$: the renderer need not distinguish every possible latent distribution, only the pretrained latent preimage that produces $q_0(\cdot|s)$. A simple stronger sufficient condition is distributional injectivity of the renderer. In finite alphabets, if the matrix $W_s(x|u)=q_{\rm r}(x|u,s)$ has full input-row rank for every relevant $s$ (which requires $|\mathcal U|\le|\mathcal X|$), then the map $\mu\mapsto\mu W_s$ is injective and the UP condition holds. 

\begin{lemma}[Latent-distribution identification under UP]
\label{lem:up-latent}
Under the UP condition, every conditionally exact code satisfies, after averaging the secret key and stochastic encoder randomness,
\begin{equation}
p(u^n|s^n,m)
=
\prod_{i=1}^n q_{\rm lat}(u_i|s_i)
\label{eq:latent-product}
\end{equation}
for every message $m$ and every relevant context sequence $s^n$.
\end{lemma}

Lem.~\ref{lem:up-latent} shows that, under UP, every exact code induces the pretrained latent distribution $q_{\rm lat}$ for each message and context sequence. Thus exact output preservation also enforces exact latent-distribution preservation, which enables the strengthened converse in the next subsection. The proof is given in App.~\ref{app:up}.

\subsection{Capacity Region and Explicit Latent Preservation}

\begin{theorem}[Finite-alphabet rate--key capacity region under UP]
\label{thm:up}
Suppose $q_{\rm r}$ satisfies the UP condition relative to $q_0$. Then the capacity region is the closure of the inner bound in Thm.~\ref{thm:finite}.
\end{theorem}

Thm.~\ref{thm:finite} already gives achievability. Lem.~\ref{lem:up-latent} upgrades exact preservation of the public product distribution to exact preservation of the key-averaged latent product distribution. Details are given in App.~\ref{app:up}.

We note that the random-projection example showing that the outer-bound is loose in Rem.~\ref{rem:finite-converse-tightness} does not satisfy the UP condition. Prop.~\ref{prop:strict-converse-example} also contains a distributionally injective specialization for which the original outer bound in Thm.~\ref{thm:finite} remains loose; thus UP enables the strengthened converse here rather than making the original outer bound tight.

\begin{corollary}[Explicit latent preservation]
\label{cor:latent-preserving}
Relative to our primary problem, consider the narrower problem that additionally requires~\eqref{eq:latent-product} with the pretrained latent law $q_{\rm lat}(u|s)$ for every message and relevant context sequence $s^n$. For any fixed renderer, not necessarily satisfying UP, the rate--key capacity region is the closure of the inner bound in Thm.~\ref{thm:finite}, restricted to choices satisfying $\sum_v p(v,u|s)=q_{\rm lat}(u|s)$. Its capacity is given by~\eqref{eq:up-capacity} under the same marginal and key-budget constraints.
\end{corollary}
\begin{proof}
The latent-synthesis and encoder-local exactification in App.~\ref{app:finite-ach} achieve the prescribed latent product distribution. For the converse, the product distribution~\eqref{eq:up-full-product} used in App.~\ref{app:up} now follows directly from the imposed latent constraint and the memoryless renderer, rather than from UP. The remaining key-accounting, single-letterization, and auxiliary-erasure arguments are unchanged.
\end{proof}
This additional operational constraint gives a direct interpretation of distribution-preserving latent sampling~\cite{Yang2024GaussianShading}. It preserves the prescribed input statistics even when the output distribution does not identify them.

\subsection{Equivalent Capacity Characterization}

The capacity admits a simpler equivalent form. Under UP, the latent marginal is fixed to $q_{\rm lat}$. Moreover, any admissible auxiliary with $I(V;U|S)>R_k$ can be independently erased to reduce its latent-coordination cost to $R_k$ without decreasing the best feasible payload; see App.~\ref{app:up}. Consequently,
\begin{equation}
C(R_k)
=
\max_{p(v,u|s)}
\bigl[I(V;X)-I(V;S)\bigr].
\label{eq:up-capacity}
\end{equation}
The maximization is over conditional pmfs $p(v,u|s)$ satisfying
\[
\sum_v p(v,u|s)=q_{\rm lat}(u|s),
\qquad
I(V;U|S)\le R_k.
\]

Full physical-latent identification is stronger than necessary for the Gaussian model below. There, the output identifies the renderer-visible component, and the remaining physical latent can be supplied privately.

\section{Jointly Gaussian Latent-Space Watermarking}
\label{sec:gaussian}

\draftEightBegin
We now extend the fixed-renderer model to the general jointly Gaussian vector setting. Classical scalar and vector Gaussian watermarking games characterize embedding under covertext modification and attack-distortion constraints~\cite{CohenLapidoth2002Gaussian,CohenLapidoth2001Vector}. Our Gaussian problem has a different architecture: the latent-to-output renderer is fixed, the released sample must preserve its prescribed conditional distribution exactly, and the shared secret key has a finite entropy rate. These constraints make the renderer-visible component of the latent, rather than the full physical latent itself, the natural Gaussian object. We begin by identifying that component.
\draftEightEnd

\subsection{Renderer-Visible Innovation}
\label{sec:gaussian-visible}

Let the pretrained-system joint law $q_{\rm pre}(s,u,x)$ be jointly Gaussian, with $S\in\mathbb R^{d_s}$, $U\in\mathbb R^{d_u}$, and $X\in\mathbb R^d$, and assume the nondegenerate renderer residual covariance
$\Sigma_{X|S,U}\succ0.$
To define the conditional means and covariance matrices associated with the renderer, we first consider the pretrained-system joint distribution $q_{\rm pre}$, which characterizes the distribution under which the renderer was trained. Define
\begin{equation}
T
\triangleq
\EE[X|S,U]-\EE[X|S].
\label{eq:visible-innovation}
\end{equation}
Thus, $T$ captures the renderer-visible effect of the latent beyond the context $S$.

Define also the irreducible renderer residual
\begin{equation}
Z
\triangleq
X-\EE[X|S,U].
\label{eq:renderer-residual}
\end{equation}
Then
\begin{equation}
X=\EE[X|S]+T+Z.
\label{eq:gaussian-decomposition}
\end{equation}
For jointly Gaussian variables, $T\perp S$ and $Z\perp(S,U)$. Moreover,
\begin{equation}
\Sigma_T
=
\Sigma_{X|S}-\Sigma_{X|S,U},
\quad
\Sigma_Z
=
\Sigma_{X|S,U}.
\label{eq:gaussian-covariances}
\end{equation}
Thus, once $S$ is fixed, changing $U$ can alter the conditional output law only through the conditional mean, and $T$ records exactly this renderer-visible effect.

A linear Gaussian example makes the interpretation concrete.  Let
\begin{equation}
X=AS+BU+Z,
\label{eq:linear-gaussian}
\end{equation}
where $A$ and $B$ are deterministic matrices,  $Z$ is the independent Gaussian renderer noise, and $U\perp S$ for simplicity.  Then $\EE[X|S]=AS$ and $\EE[X|S,U]=AS+BU$, so
\begin{equation}
T=BU.
\label{eq:TBU}
\end{equation}
The renderer exposes the effect $BU$, not the raw physical latent $U$.  That is, watermark information is useful only through latent variation that the pretrained mapping makes visible at the output.  Latent directions in the null space $\ker B$ are annihilated by $B$: they cannot change the released sample and therefore cannot carry watermark information to a detector that observes the output.

This characterization suggests a practical inductive bias for Gaussian latent-space watermarking applications; see, e.g., \cite{Yang2024GaussianShading,Yang2026FARI}. Distribution-preserving latent sampling can maintain the generator's ordinary input statistics, but reliable watermark recovery additionally requires embedding information in latent variation that remains visible through the fixed generator. In the Gaussian model, $T$ identifies exactly this renderer-visible component. Thus the theory suggests favoring latent representations or directions that are strongly expressed at the generated output, rather than treating all distribution-preserving latent variation as equally useful.

\subsection{Reduction to the Effective Gaussian Renderer}

Equation~\eqref{eq:gaussian-decomposition} and Gaussianity imply
\begin{equation}
q_{\rm pre}(x|s,u)=q_{\rm pre}(x|s,t),
\label{eq:effective-markov}
\end{equation}
or equivalently $U-(S,T)-X$.  Once $(S,T)$ is given, the remaining variation in $U$ does not affect the renderer output distribution.

\begin{lemma}[Effective-latent reduction]
\label{lem:effective-latent}
Let $q_{\rm pre}(s,u,x)$ be the jointly Gaussian pretrained-system law. Any effective-latent code producing $T^n(S^n,M,K)$, with each $(S_i,T_i)$ lying almost surely in the support of $(S,T)$, can be lifted to a physical-latent code $U^n$ that induces the same output distribution as the effective-latent code and has the same decoding performance. The lift uses only encoder-private randomness and consumes no additional secret key. Hence, the reduction to $T$ is operationally lossless.
\end{lemma}
App.~\ref{app:gaussian} gives the block deconvolution and lifting details.

\subsection{The Scalar Case: Context Adaptation and Coordination}
\label{sec:gaussian-scalar}

We first consider the scalar case to develop intuition, and then extend the analysis to the vector setting in Sec.~\ref{sec:gaussian-vector}. Consider the scalar renderer
\begin{equation}
X=aS+U+Z,
\label{eq:scalar-gaussian}
\end{equation}
where $U\sim\mathcal N(0,\sigma_U^2)$ and $Z\sim\mathcal N(0,\sigma_Z^2)$ are independent of each other and of the Gaussian context $S$. Here $T=U$ and $\gamma=\sigma_U^2/\sigma_Z^2$. We assume $\sigma_U^2>0$ and $\sigma_Z^2>0$ to avoid degenerate cases. The zero-visible-variance case has zero capacity and is handled separately. We now analyze the scalar Gaussian case in parallel with the finite-alphabet treatment of Sec.~\ref{sec:finite}, exposing the corresponding coding and coordination requirements.

A Gaussian auxiliary performs two distinct operations:
\begin{equation}
V=U+\alpha aS+N_V,\quad
\alpha=\frac{\gamma}{1+\gamma},
\label{eq:scalar-auxiliary}
\end{equation}
where the independent Gaussian noise has variance $\sigma_{N_V}^2=\sigma_U^2/(2^{2r}-1)$ for $r>0$. The affine term adapts to the hidden context in the Costa sense~\cite{Costa1983DirtyPaper}. Indeed, $V-\alpha X=(1-\alpha)U-\alpha Z+N_V$ is independent of $(S,X)$ by Gaussianity and the choice of $\alpha$, so $I(V;S|X)=0$. Separately, the noise level sets the coordination cost $r=I(V;U|S)$.

Define
\begin{equation}
f_\gamma(r)=\frac12\log\frac{1+\gamma}{1+\gamma2^{-2r}}.
\label{eq:scalar-information-curve}
\end{equation}
Gaussian evaluation then gives
\begin{align}
I(V;X)-I(V;S)&=I(V;X|S)=f_\gamma(r),\\
I(V;U|S,X)&=r-f_\gamma(r).
\end{align}
The $r=0$ endpoint uses a constant auxiliary. Recall from~\eqref{eq:excess-multicoding} that $t$ denotes the excess auxiliary-codebook rate above the covering requirement. The resource constraints~\eqref{eq:net-packing}--\eqref{eq:net-coordination} then become $R+t<f_\gamma(r)$ and $R_k+t>r$, with $t>0$. Taking closure and eliminating $t$ gives
\begin{equation}
R\le f_\gamma(r)-[r-R_k]^+.
\label{eq:scalar-resource-ceiling}
\end{equation}

The derivative reveals which representation to choose:
\begin{equation}
0<f_\gamma'(r)=\frac{\gamma2^{-2r}}{1+\gamma2^{-2r}}<1.
\label{eq:scalar-slope}
\end{equation}
Below the available key budget, more coordination improves the payload ceiling. Above it, every extra coordination bit consumes one bit of packing through $t$ but gains less than one bit in that ceiling. Thus~\eqref{eq:scalar-resource-ceiling} is maximized at $r=R_k$, giving
\begin{equation}
C(R_k)=\frac12\log\frac{1+\gamma}{1+\gamma2^{-2R_k}}.
\label{eq:scalar-capacity}
\end{equation}
The Gaussian converse in Thm.~\ref{thm:gaussian} below establishes that this achieved value is capacity. The optimum has $t=0$ in closure; the context-adaptation part $I(V;S)$ of multicoding need not vanish. The endpoints are $C(0)=0$ and $C(\infty)=\frac12\log(1+\gamma)$.

This separates adaptation to the state from selection of the amount of latent detail. A richer representation is not always better under finite key. Across several modes, the remaining question is where that coordination is most productive.

\subsection{Vector Gaussian Finite-Key Capacity}
\label{sec:gaussian-vector}

Define the covariance of the whitened renderer-visible latent $T$ as
\begin{equation}
\Gamma
\triangleq
\Sigma_Z^{-1/2}\Sigma_T\Sigma_Z^{-1/2}
=
\Sigma_{X|S,U}^{-1/2}
\bigl(\Sigma_{X|S}-\Sigma_{X|S,U}\bigr)
\Sigma_{X|S,U}^{-1/2}.
\label{eq:Gamma}
\end{equation}
Let $\gamma_1,\ldots,\gamma_d$ denote the eigenvalues of $\Gamma$.  After whitening the renderer noise and diagonalizing $\Gamma$, the eigenvectors define orthogonal renderer-visible directions.  We refer to the resulting decoupled scalar Gaussian subchannels as \emph{renderer-visible modes}; the corresponding eigenvalue $\gamma_j$ is the effective SNR of mode $j$.

\begin{theorem}[Jointly Gaussian finite-key capacity]
\label{thm:gaussian}
Let $S\in\mathbb R^{d_s}$, $U\in\mathbb R^{d_u}$, and $X\in\mathbb R^d$ have the jointly Gaussian pretrained-system law $q_{\rm pre}$ with $\Sigma_{X|S,U}\succ0$, $\gamma_1,\ldots,\gamma_d$ be the eigenvalues of the renderer-visible matrix $\Gamma$ in~\eqref{eq:Gamma}, and $r_j\ge0$ denote the secret-key rate allocated to renderer-visible mode $j$. Then
\begin{equation}
C(R_k)
=
C_{\ctx}(R_k)
=
\max_{\substack{r_j\ge0\\\sum_{j=1}^d r_j=R_k}}
\sum_{j=1}^d
\frac12\log
\frac{1+\gamma_j}
{1+\gamma_j2^{-2r_j}}.
\label{eq:gaussian-capacity}
\end{equation}
\end{theorem}

The theorem characterizes finite-key capacity and shows that revealing the context to the detector does not increase it. The condition $\Sigma_{X|S,U}\succ0$ excludes deterministic renderers and noiseless output directions. Such cases require separate support analysis; singular covariance alone does not imply infinite capacity, because a noiseless direction may contain no renderer-visible latent variation.

The proof is deferred to App.~\ref{app:gaussian}.

\subsection{Renderer-Visible Modes and Key Allocation}

In the linear example~\eqref{eq:TBU}, directions eliminated by $B$ generate zero renderer-visible modes and contribute no watermark rate.

At unlimited key,
\begin{equation}
C(\infty)
=
\frac12\log
\frac{\det\Sigma_{X|S}}{\det\Sigma_{X|S,U}}
=
I(U;X|S)
=
I(T;X|S).
\label{eq:gaussian-unlimited}
\end{equation}
\draftEightBegin
Thus unlimited secret key cannot make renderer-invisible latent variation useful. A higher key rate can coordinate visible modes more completely, but it cannot create an output effect in a latent direction that the pretrained renderer does not expose at the output.
\draftEightEnd

For finite key rate, the optimal allocation has the log-threshold form
\begin{equation}
r_j^\star
=
\frac12\left[\log\frac{\gamma_j}{\lambda}\right]^+,
\label{eq:key-allocation}
\end{equation}
for $\gamma_j>0$, with $r_j^\star=0$ when $\gamma_j=0$; $\lambda$ is selected so that the key constraint is active when appropriate.  The latent coordination budget from Sec.~\ref{sec:up} therefore acquires a geometric meaning: with scarce key, coordination is allocated first to the renderer-visible modes with the largest effective SNRs.

\subsection{Numerical Illustration of Renderer Geometry}

Figure~\ref{fig:gaussian-geometry} evaluates~\eqref{eq:gaussian-capacity} for two renderer-visible mode-SNR profiles with the same total effective SNR,
\begin{equation}
\boldsymbol\gamma_{\rm flat}=(2,2,2,2),
\quad
\boldsymbol\gamma_{\rm aniso}=(7,0.5,0.3,0.2),
\end{equation}
so that $\sum_j\gamma_j=8$ in both cases, and compares optimal key allocation with uniform allocation for the anisotropic renderer.  With scarce secret key, the optimally allocated anisotropic profile is superior because the coordination budget can be concentrated on its strongest mode.  The two optimally allocated curves cross at approximately $R_k=2.36$ bits/sample, after which the flat profile is superior.  At unlimited key,~\eqref{eq:gaussian-unlimited} gives approximately $3.17$ bits/sample for the flat profile and $2.11$ bits/sample for the anisotropic profile; under a fixed sum $\sum_j\gamma_j$, this large-key ordering follows from the concavity of $\log(1+\gamma)$.  Thus renderer anisotropy is neither uniformly beneficial nor uniformly harmful: it trades low-key efficiency against high-key aggregate watermark capacity.  The uniform-allocation curve further shows the cost of ignoring the renderer-induced mode geometry.

\begin{figure}[t]
\centering
\includegraphics[width=0.58\linewidth]{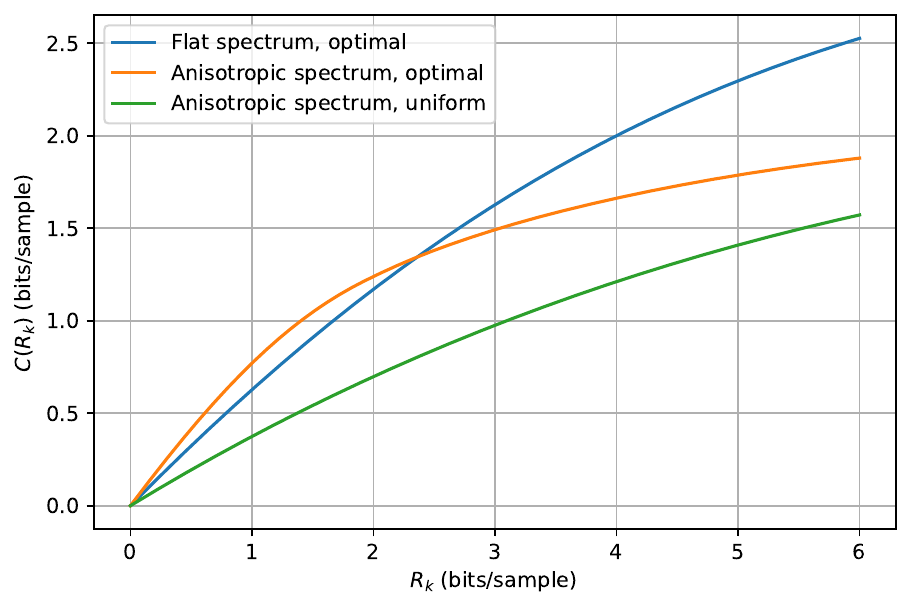}
\caption{Jointly Gaussian capacity versus secret-key rate for two renderer-visible mode-SNR profiles with equal total effective SNR $\sum_j\gamma_j=8$.  The anisotropic uniform-allocation curve uses $r_j=R_k/4$; the other two curves use the optimal allocation~\eqref{eq:key-allocation}.}
\label{fig:gaussian-geometry}
\end{figure}

\section{Robustness to Same-Generator Regeneration}
\label{sec:robust}

The preceding section studied how watermark information is retained through a generator. From a robustness perspective, the generator itself can already be viewed as a transformation that may attenuate or erase watermark-bearing variation. In this section, we take this viewpoint one step further and ask what happens when an adversary deliberately invokes the generator again to produce a fresh realization of the same underlying content. Would a watermark be robust against such attacks? We first formalize this regeneration attack and its reduction to an effective renderer, and then study compound robustness and repeated regeneration in the scalar Gaussian model.

\subsection{Regeneration Model and Preservation}
\label{sec:model-attack}

We now formulate the regeneration attack model. The attack is modeled as a state-dependent memoryless channel with the following properties.

The attacker maps a previously generated sample $X$, together with its semantic context $S$, to a new generation $Y$. In the regeneration attack considered here, this mapping is realized by reusing the pretrained generator. At blocklength $n$,
\begin{equation}
a(y^n|x^n,s^n)
=
\sum_{u_{\rm a}^n}
a(u_{\rm a}^n|x^n,s^n)
q_{\rm r}(y^n|u_{\rm a}^n,s^n),
\label{eq:attack-render-block}
\end{equation}
where $a(u_{\rm a}^n|x^n,s^n)$ specifies the attacker's choice of latent input. We assume this attack-latent mapping is memoryless,
\begin{equation}
a(u_{\rm a}^n|x^n,s^n)
=
\prod_{i=1}^n a(u_{{\rm a},i}|x_i,s_i).
\label{eq:attack-latent-block}
\end{equation}
Since the pretrained renderer is also memoryless, the resulting regeneration channel satisfies
$
a(y^n|x^n,s^n)
=
\prod_{i=1}^n a(y_i|x_i,s_i).
$
We allow the attack remapping to depend on the semantic context $S$, while the attack remains independent of the watermark message and secret key except through its permitted observations.

Under attack, the same watermark detector is applied to the regenerated sample and the secret key,
$\hat M=g_n(Y^n,K).$
Reliability under an admissible one-pass attack requires
$
\PP\{g_n(Y^n,K)\neq M\}\to0.
$
For compound robustness, the same watermark code and detector must satisfy this requirement uniformly over the specified attack family.

The attacked sample must itself remain a valid ordinary generation.  Operationally, for every watermark message $m$,
\begin{equation}
p(y^n|s^n,m)
=
\prod_{i=1}^n q_0(y_i|s_i)
\label{eq:attack-exact-output}
\end{equation}
for almost every $s^n$ under $q_0(s^n)$; on finite alphabets, equality is imposed whenever $q_0(s^n)>0$.  This restriction prevents the attacker from winning by simply destroying the sample: the regenerated output must continue to have the desired conditional output distribution.  Finally, each regeneration pass is restricted by the block-average fidelity constraint
\begin{equation}
\frac1n\sum_{i=1}^n \EE\bigl[d_{\rm a}(X_i,Y_i)\bigr]
\le D.
\label{eq:attack-distortion-model}
\end{equation}
We next derive the corresponding single-letter family and distinguish a fixed attack from uniform reliability over the family.

\subsection{Renderer-Realizable Regeneration and Fixed-Attack Reduction}

\begin{definition}[Renderer-realizable attack]
\label{def:renderer-realizable}
For each context $s$, let $A_s$ denote the transition operator of the attack kernel $a(y|x,s)$. On finite alphabets, $A_s$ is the stochastic matrix with entries $[A_s]_{x,y}=a(y|x,s)$; on continuous alphabets, the corresponding Markov-kernel interpretation is used. We write $A=(A_s)_s$ for the context-indexed attack. We say that $A_s$ is \emph{renderer-realizable} if there exists an attack-latent kernel $a(u_{\rm a}|x,s)$ such that
\begin{equation}
a(y|x,s)
=
\sum_{u_{\rm a}}
a(u_{\rm a}|x,s)q_{\rm r}(y|u_{\rm a},s).
\label{eq:renderer-realizable}
\end{equation}
\end{definition}

\draftEightBegin
Operationally, the attacker first maps the released sample $X$ to an attack latent $U_{\rm a}$ through $a(u_{\rm a}|x,s)$, after which the unchanged pretrained renderer $q_{\rm r}(y|u_{\rm a},s)$ produces $Y$. Equation~\eqref{eq:renderer-realizable} is the resulting single-letter composition of those two kernels and is precisely the channel induced by~\eqref{eq:attack-render-block}--\eqref{eq:attack-latent-block}.
\draftEightEnd

\begin{lemma}[Fixed-attack reduction]
\label{lem:attack-reduction}
Suppose the incoming watermarked sample satisfies~\eqref{eq:exact-output}.  For a memoryless, message- and key-independent attack $A$, the post-attack exact output requirement~\eqref{eq:attack-exact-output} is equivalent to
\begin{equation}
q_0(\cdot|s)A_s=q_0(\cdot|s)
\label{eq:q0-stationary}
\end{equation}
for every relevant $s$.  The end-to-end channel from $U$ to $Y$ is the effective renderer
\begin{equation}
q_{\rm r}^{A}(y|u,s)
=
\sum_x q_{\rm r}(x|u,s)a(y|x,s),
\label{eq:effective-renderer}
\end{equation}
which has the same desired output distribution $q_0(\cdot|s)$.
\end{lemma}

The proof follows by composing the incoming exact product law with the memoryless attack.  Importantly, the stationarity condition~\eqref{eq:q0-stationary} is not an additional modeling assumption: it is the single-letter consequence of requiring the regenerated block output to remain a valid ordinary generation as specified in Sec.~\ref{sec:model-attack}.

We can therefore define the exact-output-preserving renderer-realizable family
\begin{equation}
\Areg
=
\{A:\ A\text{ is renderer-realizable and }q_0A=q_0\},
\label{eq:Areg}
\end{equation}
and, for the one-letter distortion induced by the operational per-pass criterion,
$D(A)=\EE[d_{\rm a}(X,Y)],$
define
\begin{equation}
\Areg(D)
=
\{A\in\Areg:\ D(A)\le D\}.
\label{eq:AregD}
\end{equation}
The minimum distortion of any genuine output-preserving renderer-realizable regeneration is
\begin{equation}
D_{\min}
\triangleq
\inf_{A\in\Areg}D(A).
\label{eq:Dmin-general}
\end{equation}

For any fixed $A\in\Areg(D)$, Lem.~\ref{lem:attack-reduction} turns the attack problem into the base latent-space watermarking problem with detector observation $Y$ and renderer $q_{\rm r}^{A}$. Thus, the inner and outer bounds in Thm.~\ref{thm:finite} apply directly, and if $q_{\rm r}^{A}$ satisfies the UP condition relative to $q_0$, Thm.~\ref{thm:up} gives the fixed-attack capacity region. \draftEightBegin
The genuinely new issue is therefore uniformity: one code and one detector must work over an attack family.
\draftEightEnd

The resulting robustness problem is compound: the same watermark code and detector must work uniformly for every $A\in\Areg(D)$.  In general finite alphabets, such a capacity does not follow by simply inserting $\inf_A$ into the fixed-attack formulas.

A useful endpoint is independent resampling.  The attacker discards $X$, draws
$U_{\rm a}\sim q_{\rm lat}(\cdot|S),$
and rerenders, giving $[A_{{\rm res},s}]_{x,y}=q_0(y|s)$.  Define
\begin{equation}
D_{\rm res}
\triangleq
\EE[d_{\rm a}(X,\tilde X)],
\label{eq:Dres-general}
\end{equation}
where, conditioned on $S$, $X$ and $\tilde X$ are independent draws from $q_0(\cdot|S)$.  Then
\begin{equation}
D\ge D_{\rm res}
\quad\Longrightarrow\quad
C_{{\rm cmp},1}(D,R_k)=0.
\end{equation}
This is a guaranteed zero-capacity threshold, not necessarily the smallest distortion at which some admissible attack can kill the watermark.

\subsection{Scalar Gaussian Regeneration Geometry}

\draftEightBegin
We now specialize the regeneration model to the scalar Gaussian renderer~\eqref{eq:scalar-gaussian}, while retaining the original information pattern: the random context $S^n$ is known noncausally to the watermark encoder, the detector does not observe $S^n$, and the attacker may use $S$ as allowed in Sec.~\ref{sec:model-attack}.  Define the residual
\draftEightEnd
\begin{equation}
\bar X\triangleq X-aS=U+Z,
\label{eq:Xbar}
\end{equation}
where 
$
\sigma_{\bar X}^2
=
\sigma_U^2+\sigma_Z^2
$
\draftEightBegin
and $\sigma_S^2=\operatorname{Var}(S)$.  The residual is an analytical device; the detector continues to observe only the attacked public sample and the secret key.
\draftEightEnd
A same-generator rerender takes the form
$Y=aS+U_{\rm a}+Z_{\rm a},$
\draftEightBegin
with fresh renderer noise $Z_{\rm a}\sim\mathcal N(0,\sigma_Z^2)$.  Exact preservation of the conditional Gaussian output law requires, for every relevant $s$,
\begin{equation}
Y|S=s\sim\mathcal N(as,\sigma_{\bar X}^2).
\label{eq:gaussian-attack-exact}
\end{equation}
Equivalently, with $\bar Y\triangleq Y-aS$, we have $\bar Y|S=s\sim\mathcal N(0,\sigma_{\bar X}^2)$ for every $s$, and hence $\bar Y\perp S$.  Gaussian deconvolution further gives
$U_{\rm a}|S=s\sim\mathcal N(0,\sigma_U^2)$ for every $s$.

Under squared-error attack distortion, the context shift cancels and
$\EE[(X-Y)^2]=\EE[(\bar X-\bar Y)^2]$.  Therefore the minimum genuine rerendering distortion remains
\draftEightEnd
$
D_{\min}
=
2\sigma_{\bar X}^2
-2\sigma_{\bar X}\sigma_U.
$
A stochastic renderer necessarily injects fresh noise on every genuine rerender, so $D_{\min}>0$ whenever $\sigma_Z^2>0$; if $\sigma_Z^2=0$, then $D_{\min}=0$.

For the canonical Gaussian regeneration family, let
$U_{\rm a}=\beta\bar X+N_{\rm a},$
\draftEightBegin
where independent Gaussian $N_{\rm a}$ is chosen so that $U_{\rm a}$ has the required marginal.  The attacked residual satisfies
\draftEightEnd
\begin{equation}
\bar Y=\beta\bar X+W,
\label{eq:canonical-attack}
\end{equation}
where $W$ is independent Gaussian with variance $\sigma_{\bar X}^2(1-\beta^2)$.  The distortion is
$D(\beta)=2\sigma_{\bar X}^2(1-\beta),$
so
\begin{equation}
\beta(D)=1-\frac{D}{2\sigma_{\bar X}^2}.
\label{eq:betaD}
\end{equation}
For $D_{\min}\le D\le D_{\rm res}$, where $D_{\rm res}=2\sigma_{\bar X}^2$, the canonical remapping is renderer-realizable exactly when $\operatorname{Var}(N_{\rm a})=\sigma_U^2-\beta^2\sigma_{\bar X}^2\ge0$. In the range above, this is $0\le\beta\le\sigma_U/\sigma_{\bar X}$. Thus the lower endpoint is the least-distorting genuine rerender and the upper endpoint is independent resampling. Below $D_{\min}$, the genuine one-pass attack family is empty, rather than a family whose capacity is given by extrapolating the formula.

\subsection{One-Pass Gaussian Compound Capacity}

\draftEightBegin
Let $C_{{\rm cmp},1}^{\rm G}(D,R_k)$ be the largest watermark rate achieved by one code and one detector uniformly over the full one-pass scalar Gaussian family $\Areg(D)$ defined above.  The detector remains the base-model detector $g_n(Y^n,K)$ and therefore does not know $S^n$.
\draftEightEnd

\begin{theorem}[One-pass compound Gaussian capacity]
\label{thm:compound}
\draftEightBegin
For the scalar Gaussian model~\eqref{eq:scalar-gaussian} with squared-error attack distortion and the full renderer-realizable one-pass family $\Areg(D)$, for $D_{\min}\le D\le D_{\rm res}$,
\draftEightEnd
\begin{equation}
C_{{\rm cmp},1}^{\rm G}(D,R_k)
=
\frac12\log
\frac{\sigma_{\bar X}^2}
{\sigma_{\bar X}^2-
\beta(D)^2\sigma_U^2(1-2^{-2R_k})}.
\label{eq:compound-capacity}
\end{equation}
Moreover,
$
C_{{\rm cmp},1}^{\rm G}(D,R_k)=0
$
for  $D\ge D_{\rm res}$.
\draftEightBegin
The canonical Gaussian regeneration~\eqref{eq:canonical-attack} is least favorable over the full admissible one-pass family, and one common state-adapted Gaussian watermark construction achieves~\eqref{eq:compound-capacity} uniformly over that family.
\draftEightEnd
\end{theorem}

\draftEightBegin
The achievability combines the common-auxiliary compound Gel'fand--Pinsker covering/packing structure~\cite{PiantanidaShamai2010CompoundState} with the latent-synthesis and exactification requirements of the present problem.  A single attack-independent Costa-type auxiliary adapts to the random context while a universal compound decoder operates directly on the attacked sample $Y^n$, without observing $S^n$.  App.~\ref{app:compound} gives the details.  In this scalar model, the same capacity value also holds over the larger class of all memoryless state-dependent Gaussian-output-preserving attacks with distortion at most $D$, throughout $D_{\min}\le D\le D_{\rm res}$.
\draftEightEnd

\subsection{Repeated Regeneration and Watermark Survival}

\draftEightBegin
Prior work studies diffusive and regeneration attacks that degrade or remove image watermarks~\cite{An2024WAVES,Zhao2024Removable}. The result below is a distortion-parameterized capacity calculation for a specified canonical Gaussian channel sequence; it retains the random hidden-context information pattern of Thm.~\ref{thm:compound}.
\draftEightEnd

For the $\ell$-th regeneration round of the canonical Gaussian model, write
\begin{equation}
\bar X^{[\ell]}
=
\beta(D)\bar X^{[\ell-1]}+W_\ell.
\label{eq:repeated-recursion}
\end{equation}
\draftEightBegin
Every residual $\bar X^{[\ell]}$ is independent of $S$ and distributed as $\mathcal N(0,\sigma_{\bar X}^2)$, so each regenerated public sample preserves the desired conditional law exactly.  Iterating~\eqref{eq:repeated-recursion} and collecting the independent Gaussian noises gives
\begin{equation}
X^{[L]}
=
aS+\beta(D)^L U+N_L,
\label{eq:repeated-effective-renderer}
\end{equation}
where $N_L\perp(S,U)$ is Gaussian with variance $\sigma_{\bar X}^2-\beta(D)^{2L}\sigma_U^2$.
\draftEightEnd
After $L$ additional renderer uses, the original latent contribution is attenuated by $\beta(D)^L$.  The effective renderer-visible SNR is
\begin{equation}
\gamma_L(D)
=
\frac{\beta(D)^{2L}\sigma_U^2}
{\sigma_{\bar X}^2-\beta(D)^{2L}\sigma_U^2}.
\label{eq:gammaL}
\end{equation}
\draftEightBegin
The effective renderer~\eqref{eq:repeated-effective-renderer} is jointly Gaussian, so Thm.~\ref{thm:gaussian} applies directly and shows that hiding the context from the detector causes no capacity loss.  Hence the exact canonical repeated-regeneration curve is
\draftEightEnd
\begin{equation}
C_L^{\rm G}(D,R_k)
=
\frac12\log
\frac{\sigma_{\bar X}^2}
{\sigma_{\bar X}^2-
\beta(D)^{2L}\sigma_U^2(1-2^{-2R_k})}.
\label{eq:repeated-capacity}
\end{equation}
For this canonical family, $C_L^{\rm G}$ is nonincreasing in both $D$ and $L$, and if $|\beta(D)|<1$, then $C_L^{\rm G}(D,R_k)\to0$ as $L\to\infty$.

In particular, for a stochastic renderer $\sigma_Z^2>0$,
$\beta(D_{\min}) = \frac{\sigma_U}{\sigma_{\bar X}}<1,$
so even the least-distorting genuine canonical rerender eventually drives every positive target watermark rate below threshold.

The distortion $D$ is imposed \emph{per renderer reuse}.  Under the canonical family, the end-to-end distortion after $L$ passes is
$
D_{\rm end}^{(L)}
=
2\sigma_{\bar X}^2\bigl(1-\beta(D)^L\bigr),
$
which generally differs from $D$.  Thus $L$ regeneration rounds cannot be reinterpreted as one distortion-$D$ attack.

For a target rate $0<R\le C_0^{\rm G}(R_k)$, define the survival horizon, where $C_0^{\rm G}(R_k)$ denotes the unattacked scalar capacity and $L=0$ means no regeneration,
\begin{equation}
L_{\rm surv}^{\rm G}(R;D,R_k)
=
\max\{L\ge0:C_L^{\rm G}(D,R_k)\ge R\}.
\end{equation}
Let
\begin{equation}
\rho_R(R_k)
=
\frac{\sigma_{\bar X}^2}{\sigma_U^2}
\frac{1-2^{-2R}}
{1-2^{-2R_k}}.
\label{eq:rhoR}
\end{equation}
For $R_k>0$, $0<\beta(D)<1$, and $0<\rho_R(R_k)\le1$ (equivalently, the target is feasible without attack),
\begin{equation}
L_{\rm surv}^{\rm G}(R;D,R_k)
=
\left\lfloor
\frac{\ln\rho_R(R_k)}{2\ln\beta(D)}
\right\rfloor.
\label{eq:survival-horizon}
\end{equation}
A value of zero means that the target is achievable before attack but no positive number of passes preserves it. If $\rho_R(R_k)>1$, the target is already infeasible without regeneration and the defining set is empty. At $D=D_{\rm res}$, $\beta(D)=0$ and the survival horizon is zero for every positive feasible target. These conventions avoid extrapolating~\eqref{eq:survival-horizon} outside its stated domain.

\draftEightBegin
Thm.~\ref{thm:compound} concerns the full one-pass scalar Gaussian attack family with random context hidden from the detector. Equation~\eqref{eq:repeated-capacity} instead concerns repeated \emph{canonical Gaussian} regeneration under the same information pattern. It is not an unrestricted multiround minimax result: the canonical attack is not claimed to be least favorable over arbitrary $L>1$ attack sequences.
\draftEightEnd

\subsection{Numerical Illustration of Regeneration Survival}

Figure~\ref{fig:regeneration} plots~\eqref{eq:repeated-capacity} versus the number of regeneration attack rounds $L$ for fixed normalized per-pass distortions
\begin{equation}
\delta_D
\triangleq
\frac{D-D_{\min}}{D_{\rm res}-D_{\min}}
\in\{0,0.25,0.5,0.75,1\}.
\label{eq:deltaD}
\end{equation}
The parameters are $\sigma_Z^2=1$, $\sigma_U^2=4$, and $R_k=2$ bits/sample.  For any fixed attack severity, the plot directly shows the progressive loss of watermark rate across repeated same-generator regeneration.  The endpoints $\delta_D=0$ and $\delta_D=1$ correspond to the least-distorting genuine rerender and independent resampling, respectively.

\begin{figure}[t]
\centering
\includegraphics[width=0.58\linewidth]{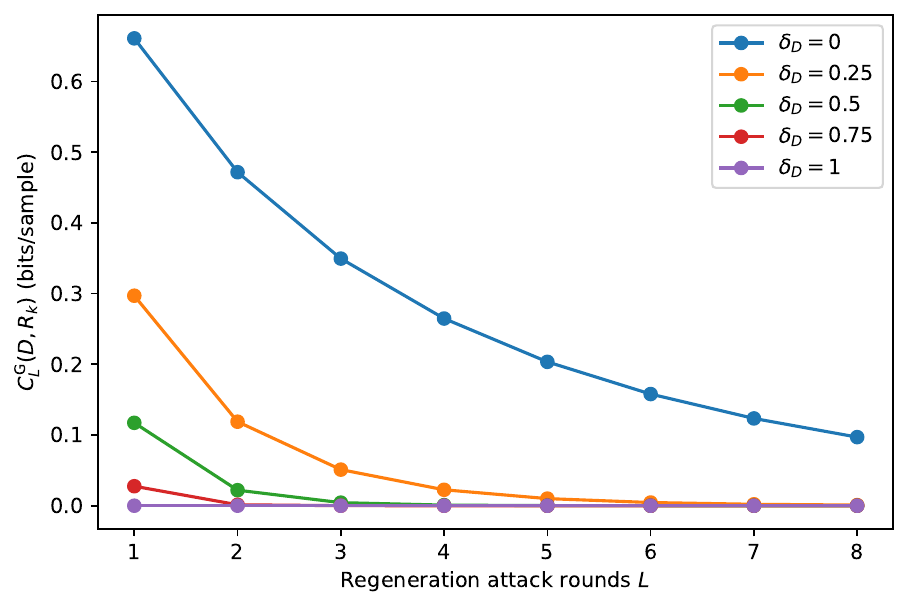}
\caption{Canonical Gaussian watermark capacity versus regeneration attack rounds $L$ for several normalized per-pass distortion levels $\delta_D$.  Parameters are $\sigma_Z^2=1$, $\sigma_U^2=4$, and $R_k=2$ bits/sample.}
\label{fig:regeneration}
\end{figure}
\FloatBarrier

\section{Discussion}
\label{sec:discussion}

We studied latent-space watermarking through a fixed pretrained generator. The encoder controls the latent input, but reliable recovery depends on how the generator transforms that input into the released sample. The Gaussian case gives a useful illustration of this distinction. Even when the latent prior is isotropic, the generator need not preserve information equally well in every direction. The optimal key allocation accounts for these differences and, when the key rate is small, concentrates the available key on the stronger modes. This suggests that watermark design should account for the generator's behavior, even when the generator itself is not modified.

\draftEightBegin
The same generator also provides a natural way to study regeneration attacks. In the scalar Gaussian model, a single state-adapted watermark code and context-blind detector can be made robust to the full admissible one-pass attack family, even though the attack itself may use the semantic context. Repeated canonical regeneration can nevertheless progressively reduce the recoverable watermark rate while preserving the desired conditional output distribution at every round. Thus, preserving the statistical properties of an ordinary generation does not by itself protect the watermark. Extending the compound minimax analysis to arbitrary multiround attack sequences, beyond the canonical Gaussian family considered here, remains an interesting direction for further work.
\draftEightEnd

\appendices

\section{Proof of Theorem~\ref{thm:finite}: Achievability}
\label{app:finite-ach}

In this appendix, we prove the achievability part of Thm.~\ref{thm:finite}.  For a fixed $p\in\mathscr P(q_0;q_{\rm r})$, we first establish the standard state-dependent synthesis and packing conditions in the native variables $(S,V,U,X)$.  We then address the step specific to the present problem: the preliminary synthesis construction only approximates the desired key-averaged latent distribution, whereas~\eqref{eq:exact-output} requires finite-block exact message-wise output preservation.  We therefore finish the proof by a conditional maximal-coupling exactification before the fixed renderer.

Fix $p\in\mathscr P(q_0;q_{\rm r})$. Its conditional latent marginal is
$
p(u|s)
\triangleq
\sum_v p(v,u|s).
$
Since $p\in\mathscr P(q_0;q_{\rm r})$,
\begin{equation}
\sum_u p(u|s)q_{\rm r}(x|u,s)
=
q_0(x|s).
\label{eq:app-induced-output}
\end{equation}
The exactification target is this selected $p(u|s)$, not necessarily $q_{\rm lat}(u|s)$; both produce the same desired output, but the general problem does not require them to agree.

Generate a random codebook whose codewords $V^n(m,\ell,k)$ are drawn independently according to $\prod_{i=1}^n p(v_i)$. Here $m$ and $k$ are the message and key indices, while $\ell\in[1:2^{n\hat R}]$ is an internal randomization index. Condition on a realization of this codebook. Given $(s^n,m,k)$, the likelihood encoder selects $\ell$ using the posterior
\begin{equation}
p(\ell|s^n,m,k)
=\frac{\prod_{i=1}^n p(s_i|v_i(m,\ell,k))}
{\sum_{\ell'}\prod_{i=1}^n p(s_i|v_i(m,\ell',k))}.
\label{eq:likelihood-encoder-posterior}
\end{equation}
If the denominator is zero, use any fixed distribution over the index set. Selecting $V^n=v^n(m,\ell,k)$ induces the likelihood-encoder distribution $p(v^n|s^n,m,k)$. The single-letter factors $p(s|v)$ are obtained from the chosen $p(s,v,u,x)$; the block likelihood encoder is not assumed to be the memoryless extension of $p(v|s)$~\cite{SongCuffPoor2016Likelihood}.

Given the selected auxiliary sequence $V^n$, generate a \emph{preliminary latent sequence} $\bar U^n$ through the memoryless extension of the latent-completion kernel:
\begin{equation}
p(\bar u^n|s^n,v^n)
=\prod_{i=1}^n p(\bar u_i|s_i,v_i).
\label{eq:memoryless-latent-completion}
\end{equation}
Each factor is the same single-letter kernel $p(u|s,v)$ evaluated at $(u,s,v)=(\bar u_i,s_i,v_i)$. We call~\eqref{eq:memoryless-latent-completion} the \emph{memoryless latent-completion kernel}. The bar distinguishes the latent produced by this preliminary construction from the final latent $U^n$: a later encoder-local correction will enforce exactness at the renderer input, before public generation. In particular,
\begin{equation}
p(\bar u^n|s^n,m,k)
=\sum_{v^n}p(v^n|s^n,m,k)p(\bar u^n|s^n,v^n).
\label{eq:preliminary-latent-encoder}
\end{equation}

For each fixed $(m,k)$, the likelihood encoder and high-probability soft covering for the channel $V\to S$ give simultaneous approximation of the source product distribution whenever
\begin{equation}
\hat R>I(V;S).
\label{eq:raw-cover-1}
\end{equation}
For each fixed message $m$, averaging over the aggregate $(k,\ell)$ subcodebook synthesizes the target pair $(S,U)$ whenever
\begin{equation}
R_k+\hat R>I(V;S,U).
\label{eq:raw-cover-2}
\end{equation}
The likelihood encoder uses the ideal posterior associated with the selected target distribution~\cite{SongCuffPoor2016Likelihood}.  Standard packing over the induced $V\to X$ channel gives reliable decoding, with $K$ available at the detector, if
\begin{equation}
R+\hat R<I(V;X).
\label{eq:raw-pack}
\end{equation}
Cuff's high-probability soft-covering theorem makes the two covering conditions hold simultaneously over the exponentially many relevant subcodebooks~\cite{Cuff2016HighProbability}, while standard random-coding packing makes the packing-good event hold with probability tending to one.  Hence their intersection is nonempty, and a deterministic codebook satisfying all three properties exists.  In particular, the preliminary latent sequence from~\eqref{eq:memoryless-latent-completion} satisfies
\begin{equation}
\max_m \delta_n(m)\to0,
\label{eq:delta-max}
\end{equation}
where
\begin{equation}
\delta_n(m)
\triangleq
\left\|
p(s^n,\bar u^n|m)
-
\prod_{i=1}^n q_0(s_i)p(\bar u_i|s_i)
\right\|_{\rm TV}.
\label{eq:delta-message}
\end{equation}

Eliminating $\hat R$ from~\eqref{eq:raw-cover-1}--\eqref{eq:raw-pack} yields
\begin{align}
R
&<I(V;X)-I(V;S),\label{eq:app-inner-rate-strict}\\
R_k
&>R+I(V;S,U|X),
\label{eq:app-inner-key-strict}
\end{align}
where the second identity uses the Markov relation $V-(S,U)-X$.

It remains to replace the preliminary approximation~\eqref{eq:delta-message} by the exact condition required by the watermarking problem.  The two joint distributions in~\eqref{eq:delta-message} have the same $S^n$ marginal.  For every $(s^n,m)$ define
\begin{equation}
\delta_n(s^n,m)
\triangleq
\left\|
p(\bar u^n|s^n,m)
-
\prod_{i=1}^np(\bar u_i|s_i)
\right\|_{\rm TV}.
\label{eq:delta-conditional}
\end{equation}
Then
\begin{equation}
\delta_n(m)
=
\sum_{s^n}
\left[\prod_{i=1}^n q_0(s_i)\right]
\delta_n(s^n,m).
\label{eq:delta-average}
\end{equation}
For each $(s^n,m)$, choose a maximal coupling between the two conditional distributions in~\eqref{eq:delta-conditional}~\cite{Thorisson1986Coupling}.  Equivalently, there exists a \emph{latent-exactification kernel} $p(u^n|s^n,m,\bar u^n)$ that uses only $(S^n,M,\bar U^n)$ and encoder-private randomness and whose output $U^n$ satisfies
\begin{equation}
p(u^n|s^n,m)
=
\prod_{i=1}^np(u_i|s_i)
\label{eq:exactified-latent}
\end{equation}
exactly, with $\PP\{U^n\ne\bar U^n|M=m\}=\delta_n(m)$. This block kernel is constructed from the key-averaged preliminary distribution and is not generally memoryless. It uses $(S^n,M,\bar U^n)$ and private randomness, but need not observe or modify $K$, so it consumes no additional secret-key rate. Exactification takes place at the latent input, leaving the physical renderer unchanged. The encoder knows the conditioning variables $(s^n,m)$ and can implement this correction locally. Wagner's near-perfect-to-perfect realism argument uses the same excess/deficit-mass correction at a reconstruction marginal~\cite[Theorem~1]{Wagner2022CommonRandomness}. In remote exact channel synthesis, an output-generating decoder generally does not know the source conditioning the target channel, so this local correction is not automatically available at no rate cost~\cite{YuTan2020Exact}.

The likelihood encoder, latent-completion kernel, and latent-exactification kernel are internal stages realizing one watermark latent encoder:
\begin{align}
&p(u^n|s^n,m,k)\nonumber\\
&\quad=\sum_{v^n,\bar u^n}p(v^n|s^n,m,k)
 p(\bar u^n|s^n,v^n)\nonumber\\
&\hspace{5em}\cdot p(u^n|s^n,m,\bar u^n).
\label{eq:final-latent-encoder-composition}
\end{align}
This expression records the side information retained at each stage; in particular, the exactification kernel still depends on the message and context. The construction does not impose this factorization on a generic code in Def.~\ref{def:watermark-code}.

Passing the exactified latent through the unchanged renderer gives, for every $(s^n,m)$,
\begin{align}
&p(x^n|s^n,m)\nonumber\\
&\quad=
\sum_{u^n}
\left[\prod_{i=1}^np(u_i|s_i)\right]
\left[\prod_{i=1}^n q_{\rm r}(x_i|u_i,s_i)\right]\nonumber\\
&\quad=
\prod_{i=1}^n q_0(x_i|s_i),
\label{eq:exactification-output}
\end{align}
where the last equality follows from~\eqref{eq:app-induced-output}.  Thus the finite-block exact output requirement~\eqref{eq:exact-output} holds message-wise.

Finally, couple the preliminary and exactified renderer outputs using the same downstream renderer randomness whenever $U^n=\bar U^n$.  If $\mathsf P_{e,n}^{\rm pre}$ and $\mathsf P_{e,n}^{\rm ex}$ denote the corresponding average decoding errors for the selected deterministic codebook, then
$\mathsf P_{e,n}^{\rm ex} \le \mathsf P_{e,n}^{\rm pre} + \max_m\delta_n(m).$
The two terms on the right tend to zero by packing and~\eqref{eq:delta-max}, respectively.  Thus the strict component inequalities in Thm.~\ref{thm:finite} are achievable. Standard coded time sharing yields the convex hull stated there. The closure operation belongs to the capacity-region definition in Sec.~\ref{sec:model}, not to the achievability statement itself.

\section{Proof of Theorem~\ref{thm:finite}: Converse}
\label{app:finite-conv}

In this appendix, we prove the converse part of Thm.~\ref{thm:finite}.  We first use the exact message-wise output constraint to obtain a product distribution and a key-accounting inequality.  We then apply the standard Csisz\'ar/Gel'fand--Pinsker single-letterization to the message and key bounds.

Exact conditional equality~\eqref{eq:exact-output} implies, for every message $m$,
\begin{equation}
p(s^n,x^n|m)
=
\prod_{i=1}^n q_0(s_i)q_0(x_i|s_i).
\label{eq:product-law-app}
\end{equation}
Hence $M\perp(S^n,X^n)$.  Let
$W\triangleq(M,K).$
By Fano's inequality,
$
H(M|X^n,K)\le n\epsilon_n,
\epsilon_n\to0.
$
To bound the key rate,
\begin{align}
nR_k
&=H(K)\nonumber\\
&\ge I(K;M,S^n|X^n)\nonumber\\
&=I(K;M|X^n)+I(K;S^n|M,X^n)\nonumber\\
&\stackrel{(a)}{\ge} nR+I(W;S^n|X^n)-n\epsilon_n,
\label{eq:key-account-app}
\end{align}
where (a) uses $I(K;M|X^n)=H(M)-H(M|X^n,K)\ge nR-n\epsilon_n$, since $M\perp X^n$, and $I(K;S^n|M,X^n)=I(W;S^n|X^n)$, since $M\perp(S^n,X^n)$ and $W=(M,K)$.

For the message rate,
\begin{align}
nR
&=H(M)\nonumber\\
&\le I(M;X^n|K)+n\epsilon_n\nonumber\\
&\le I(W;X^n)-I(W;S^n)+n\epsilon_n,
\label{eq:message-multiletter-app}
\end{align}
where $I(W;S^n)=0$ since $(M,K)$ is independent of the source sequence.  Define
\begin{equation}
V_i
\triangleq
(W,X^{i-1},S_{i+1}^n).
\label{eq:Vi-app}
\end{equation}
Using the Csisz\'ar sum identity in the standard Gel'fand--Pinsker single-letterization~\cite{ElGamalKim2011}, together with the product distribution~\eqref{eq:product-law-app},
\begin{equation}
I(W;X^n)-I(W;S^n)
=
\sum_{i=1}^n
\bigl[I(V_i;X_i)-I(V_i;S_i)\bigr].
\label{eq:csiszar-app}
\end{equation}
For the key term,
\begin{align}
I(W;S^n|X^n)
&=H(S^n|X^n)-H(S^n|W,X^n)\nonumber\\
&=\sum_{i=1}^n H(S_i|X_i)
-\sum_{i=1}^n H(S_i|W,X^n,S_{i+1}^n)\nonumber\\
&\stackrel{(a)}{\ge}
\sum_{i=1}^n
\Bigl[
H(S_i|X_i)
-H(S_i|W,X^i,S_{i+1}^n)
\Bigr]\nonumber\\
&=
\sum_{i=1}^n I(V_i;S_i|X_i).
\label{eq:key-singleletter-app}
\end{align}
where (a) follows because conditioning reduces entropy.

By the memoryless renderer law, $p(x_i|s_i,v_i,u_i)=q_{\rm r}(x_i|u_i,s_i)$. Let $Q\sim\mathrm{Unif}[1:n]$ be independent of all other variables and define $S=S_Q$, $U=U_Q$, $X=X_Q$, and $V=(V_Q,Q)$.
Then,
\begin{align}
\frac1n\sum_{i=1}^n
\bigl[I(V_i;X_i)-I(V_i;S_i)\bigr]
&= I(V;X)-I(V;S),\\
\frac1n\sum_{i=1}^n I(V_i;S_i|X_i)
&= I(V;S|X).
\end{align}
The time-shared distribution satisfies the renderer factorization by memorylessness and $p(s,x)=q_0(s,x)$ by~\eqref{eq:product-law-app}; hence it belongs to $\mathscr P(q_0;q_{\rm r})$.  Combining~\eqref{eq:message-multiletter-app}, \eqref{eq:csiszar-app}, \eqref{eq:key-account-app}, and~\eqref{eq:key-singleletter-app}, dividing by $n$, and letting $n\to\infty$ yields~\eqref{eq:outer-rate}--\eqref{eq:outer-key}.

Finally, standard support reduction preserves the required source/latent marginal together with the two information functionals, so the auxiliary alphabet may be restricted to
$|\mathcal V| \le |\mathcal S||\mathcal U|+1$
~\cite{ElGamalKim2011,CoverThomas2006}.

\subsection{A Strict-Looseness Example for the General Converse}

We close the proof of Thm.~\ref{thm:finite} by giving the counterexample referenced in Rem.~\ref{rem:finite-converse-tightness}. Let the semantic source be constant and the physical latent be $U=(U_1,U_2)\in\{0,1\}^2$, with $U_1$ and $U_2$ independent $\mathrm{Bern}(1/2)$ bits under ordinary generation.  The fixed renderer outputs $X=(F,X_{\rm b})$ according to
\begin{equation}
F\sim\mathrm{Unif}[1:2],
\quad
X_{\rm b}=U_F\oplus Z,
\quad
Z\sim\mathrm{Bern}(\delta),
\label{eq:strict-example-renderer}
\end{equation}
where $F$ and $Z$ are independent of the latent and $0<\delta<1/2$.  The prescribed desired output distribution is then $q_0(f,x_{\rm b})=\frac14$.
Let $h_2(\cdot)$ denote the binary entropy function, let $h_2^{-1}$ take values in $[0,1/2]$, let $a\star b=a(1-b)+(1-a)b$, and define
\begin{equation}
g_\delta(r)
\triangleq
1-h_2\!\left(\delta\star h_2^{-1}(1-\min\{r,1\})\right).
\label{eq:gdelta}
\end{equation}

\begin{proposition}[Strict looseness of the general converse]
\label{prop:strict-converse-example}
For the renderer~\eqref{eq:strict-example-renderer} and desired distribution $q_0(f,x_{\rm b})=\frac14$, the operational capacity and the outer bound in Thm.~\ref{thm:finite} are, respectively,
\begin{align}
C(R_k)&=g_\delta(R_k),\label{eq:strict-example-capacity}\\
\overline C(R_k)&=\min\{R_k,\,1-h_2(\delta)\}.\label{eq:strict-example-outer}
\end{align}
Consequently, for every $0<R_k<1$ and $0<\delta<1/2$,
\begin{equation}
g_\delta(R_k)
<
\min\{R_k,\,1-h_2(\delta)\}.
\end{equation}
Thus the inner construction is optimal and the gap in Thm.~\ref{thm:finite} is a converse deficiency. For this two-coordinate renderer, the diagonal and independent latent distributions give different preimages of the same $q_0$, so UP fails. Hence the example demonstrates inner-bound optimality outside the identifiable class.
\end{proposition}

\begin{proof}
For a selector sequence $f^n\in[1:2]^n$, define the selected latent bits
$
B_{f^n}^n
\triangleq
(U_{1,f_1},\ldots,U_{n,f_n}),
$
where $U_i=(U_{i,1},U_{i,2})$ denotes the physical latent at time $i$.
The selector $F^n$ is fresh renderer randomness and is independent of the encoder variables.  By the exact output requirement, conditioned on $(M=m,F^n=f^n)$ the binary output $X_{\rm b}^n$ is uniform on $\{0,1\}^n$.  The product BSC$(\delta)$ is distributionally injective, so its input must also be uniform. Writing $b^n$ for a realization of $B_{f^n}^n$,
\begin{equation}
p(b^n|m,f^n)
=
2^{-n},\quad b^n\in\{0,1\}^n.
\label{eq:strict-example-projection}
\end{equation}
Hence $B_{f^n}^n\perp M$ and $H(B_{f^n}^n)=n$.  With $W=(M,K)$,
\begin{align}
I(W;B_{f^n}^n|F^n=f^n)
&=I(K;B_{f^n}^n|M,F^n=f^n)\\
&\le H(K)\le nR_k,\\
\frac1nH(B_{f^n}^n|W,F^n=f^n)
&\ge
(1-R_k)^+.
\label{eq:strict-example-keyentropy}
\end{align}
The conditional vector form of Mrs. Gerber's lemma~\cite{WynerZiv1973Entropy}, applied to the BSC$(\delta)$, therefore gives
\begin{equation}
\frac1nH(X_{\rm b}^n|W,F^n=f^n)
\ge
h_2\!\left(\delta\star h_2^{-1}((1-R_k)^+)\right).
\end{equation}
Since exact output preservation also gives $H(X_{\rm b}^n|F^n)=n$, Fano's inequality yields
\begin{align}
nR
&\le I(M;F^n,X_{\rm b}^n|K)+n\epsilon_n\\
&\le I(W;X_{\rm b}^n|F^n)+n\epsilon_n\\
&\le n g_\delta(R_k)+n\epsilon_n,
\end{align}
with $\epsilon_n\to0$.  This proves the converse in~\eqref{eq:strict-example-capacity}.

For achievability, restrict the physical latent to the diagonal form
$U=(A,A)$ with $A\sim\mathrm{Bern}(1/2)$,
and choose a binary auxiliary $V=A\oplus N_t$ with $N_t\sim\mathrm{Bern}(t)$ independent.  Then the desired output distribution is uniform on $\{1,2\}\times\{0,1\}$, and
$I(V;U)=1-h_2(t)$ and $I(V;X)=1-h_2(\delta\star t)$.
At the full rate $R=I(V;X)$ in closure, the inner key requirement in Thm.~\ref{thm:finite} reduces to $R_k\ge I(V;U)$.  Choosing $t$ so that $1-h_2(t)=\min\{R_k,1\}$ achieves $R=g_\delta(R_k)$, proving~\eqref{eq:strict-example-capacity}.

For the two-coordinate construction, the diagonal achieving distribution preserves the released-output law but not the pretrained product latent law. Thus the example also separates output preservation from the additional latent-preservation constraint in Cor.~\ref{cor:latent-preserving}.

Finally, for every admissible auxiliary distribution in the outer bound, data processing gives
$R\le I(V;X)\le I(U;X)=1-h_2(\delta),$
where the last equality follows from the uniform desired output distribution and the fact that every renderer row has conditional entropy $1+h_2(\delta)$.  The outer key condition gives $R\le R_k$.  Conversely, the same diagonal latent with $V=A$ attains the outer ceiling, yielding~\eqref{eq:strict-example-outer}.  For $0<R_k<1$, the shaping parameter satisfies $0<t<1/2$; strict data processing through the nontrivial BSC gives $g_\delta(R_k)<R_k$, while $\delta\star t>\delta$ gives $g_\delta(R_k)<1-h_2(\delta)$.  The strict inequality follows.
\end{proof}

For comparison, consider the one-coordinate specialization obtained by removing the selector $F$ and taking a scalar latent $U\in\{0,1\}$ with renderer $X=U\oplus Z$. This renderer is a distributionally injective BSC and satisfies the UP condition in Def.~\ref{def:up}. The same capacity and general-outer-bound formulas above yield the same strict gap for $0<R_k<1$. Thus preimage ambiguity is not the cause of the general outer-bound looseness.

\section{Proof of Lemma~\ref{lem:up-latent} and Theorem~\ref{thm:up}}
\label{app:up}

Fix a relevant context value $s$ and let
$\mathcal U_s \triangleq \suppset q_{\rm lat}(\cdot|s).$
Consider the renderer matrix $W_s$ whose row indexed by $u\in\mathcal U_s$ is $q_{\rm r}(\cdot|u,s)$.  We first show that these rows are linearly independent.  Otherwise there would exist a nonzero real vector $z$, supported on $\mathcal U_s$, such that
\begin{equation}
\sum_u z(u)q_{\rm r}(x|u,s)=0
\quad
\text{for every }x.
\label{eq:up-row-dependence}
\end{equation}
Since every renderer row sums to one, summing~\eqref{eq:up-row-dependence} over $x$ gives $\sum_u z(u)=0$.  Because $q_{\rm lat}(u|s)>0$ on $\mathcal U_s$, sufficiently small perturbations
$
q_{\rm lat}(\cdot|s)\pm\epsilon z
$
remain probability distributions and both render to $q_0(\cdot|s)$, contradicting Def.~\ref{def:up}.  Hence $W_s$ has full row rank on $\mathcal U_s$.

We now prove the block identification in Lem.~\ref{lem:up-latent}.  Fix a message $m$ and a relevant context sequence $s^n$.  For each coordinate $i$, define the key-averaged latent marginal
$\mu_i(u) \triangleq p(u_i|s^n,m).$
Marginalizing the exact output constraint~\eqref{eq:exact-output} to coordinate $i$ gives
$
\sum_u \mu_i(u)q_{\rm r}(x|u,s_i)
=
q_0(x|s_i)
$
for every $x$.
Thus $\mu_i\in\mathcal P_{s_i}(q_0;q_{\rm r})$. Recall that $\mathcal P_s(q_0;q_{\rm r})$ is the preimage set defined in~\eqref{eq:preimage-set}; by the UP condition in Def.~\ref{def:up}, this set contains only the pretrained latent law $q_{\rm lat}(\cdot|s)$.  Hence
$
\mu_i(u)
=
q_{\rm lat}(u|s_i).
$
Consequently, the full key-averaged conditional distribution $p(u^n|s^n,m)$ is supported on $\prod_{i=1}^n\mathcal U_{s_i}$.

On this product support, let $W_{s^n}$ denote the Kronecker product of the renderer matrices $W_{s_1},\ldots,W_{s_n}$.  Because each $W_{s_i}$ has full row rank, $W_{s^n}$ also has full row rank.  Hence the product renderer is injective on probability distributions supported on $\prod_i\mathcal U_{s_i}$.  The operational key-averaged latent distribution satisfies
\begin{align}
&\sum_{u^n}
p(u^n|s^n,m)
\prod_{i=1}^n q_{\rm r}(x_i|u_i,s_i)=
\prod_{i=1}^n q_0(x_i|s_i)
\label{eq:up-operational-preimage}
\end{align}
by exact output preservation.  The ordinary product latent distribution satisfies the same equation:
\begin{align}
\sum_{u^n}
\left[\prod_{i=1}^n q_{\rm lat}(u_i|s_i)\right]
\left[\prod_{i=1}^n q_{\rm r}(x_i|u_i,s_i)\right] =
\prod_{i=1}^n q_0(x_i|s_i).
\label{eq:up-ordinary-preimage}
\end{align}
Injectivity of $W_{s^n}$ therefore forces
\begin{equation}
p(u^n|s^n,m)
=
\prod_{i=1}^n q_{\rm lat}(u_i|s_i),
\end{equation}
which proves Lem.~\ref{lem:up-latent}.

We next prove the strengthened converse in Thm.~\ref{thm:up}, only highlighting the key differences from the converse of Thm.~\ref{thm:finite}.  To that end, Lem.~\ref{lem:up-latent} and the memoryless renderer give, for every message $m$,
\begin{equation}
p(s^n,u^n,x^n|m)
=
\prod_{i=1}^n
q_{\rm pre}(s_i,u_i,x_i),
\label{eq:up-full-product}
\end{equation}
which also implies $M\perp(S^n,U^n,X^n)$.  Repeat the key-accounting steps leading to~\eqref{eq:key-account-app}, with $(S^n,U^n)$ in place of $S^n$, which gives
\begin{equation}
H(K)
\ge
H(M)+I(W;S^n,U^n|X^n)-n\epsilon_n.
\label{eq:up-key-account}
\end{equation}
Using~\eqref{eq:up-full-product}, we have
\begin{align}
I(W;S^n,U^n|X^n)\geq\sum_{i=1}^n I(V_i;S_i,U_i|X_i),
\label{eq:up-key-singleletter}
\end{align}
where $V_i$ is defined in~\eqref{eq:Vi-app}.  By the same time-sharing steps as in App.~\ref{app:finite-conv},
$R_k \ge R+I(V;S,U|X),$
while the message-rate converse is unchanged.  This matches the closure of the achievable component of Thm.~\ref{thm:finite}, proving the capacity region in Thm.~\ref{thm:up}.

It remains to prove the equivalent capacity form~\eqref{eq:up-capacity}.  For any admissible conditional pmf $p(v,u|s)$ whose $U$-marginal is the pretrained latent law $q_{\rm lat}(u|s)$, define
$A \triangleq I(V;X)-I(V;S)$ and $J \triangleq I(V;U|S)$.
Since $V-(S,U)-X$,
\begin{equation}
J-A
=
I(V;S,U|X)
\ge0.
\label{eq:AJ-difference}
\end{equation}
For this fixed conditional pmf, Thm.~\ref{thm:up} permits
\begin{equation}
R
\le
\min\{A,\,R_k+A-J\}.
\label{eq:component-key-capacity}
\end{equation}
If $J\le R_k$, the key constraint is inactive and the minimum in~\eqref{eq:component-key-capacity} is $A$.  Suppose instead that $J>R_k$.  Let
$\theta \triangleq \frac{R_k}{J},$
and let $Q\sim\mathrm{Bern}(\theta)$ be independent of $(S,V,U,X)$.  Define a new auxiliary $V'$ by retaining $V$ when $Q=1$ and replacing it by a fixed erasure symbol when $Q=0$.  This independent erasure leaves the pretrained latent law $q_{\rm lat}(u|s)$ unchanged and gives
\begin{align}
I(V';U|S)
&=
\theta J
=
R_k,\label{eq:erased-key}\\
I(V';X)-I(V';S)
&=
\theta A.
\label{eq:erased-rate}
\end{align}
Moreover,
\begin{align}
\theta A-(R_k+A-J)
&=(1-\theta)(J-A) \ge0
\label{eq:erasure-dominates}
\end{align}
by~\eqref{eq:AJ-difference}.  Thus any nonnegative rate available from a conditional pmf with $J>R_k$ can be matched or exceeded by another admissible auxiliary satisfying $I(V';U|S)\le R_k$.  Therefore there is no loss in restricting the maximization to
$I(V;U|S)\le R_k,$
which proves~\eqref{eq:up-capacity}.

\section{Proof of Lemma~\ref{lem:effective-latent} and Theorem~\ref{thm:gaussian}}
\label{app:gaussian}

In this appendix, we prove the effective-latent reduction in Lem.~\ref{lem:effective-latent} and the vector-Gaussian capacity in Thm.~\ref{thm:gaussian}.  

\subsection{Effective-Latent Reduction and Exact Deconvolution}

Let
$\mu(S)\triangleq\EE[X|S],$
so that the Gaussian decomposition in~\eqref{eq:gaussian-decomposition} is $X=\mu(S)+T+Z$, where $T\perp S$ and $Z\perp(S,U)$, with $\Sigma_Z=\Sigma_{X|S,U}\succ0$.  We first prove Lem.~\ref{lem:effective-latent}.  Any physical-latent watermark encoder producing $U_i$ also produces
$T_i = \EE[X_i|S_i,U_i]-\EE[X_i|S_i]$
deterministically, where these conditional-mean functions are fixed by $q_{\rm pre}$, and the conditional renderer distribution depends on $(S_i,U_i)$ only through $(S_i,T_i)$. That is, conditioned on $S_i$, any change in $U_i$ that does not change $T_i$ does not affect the conditional renderer distribution.

Conversely, suppose an effective-latent watermark construction specifies $T^n(S^n,M,K)$ with each $(S_i,T_i)$ lying almost surely in the support of the ordinary $(S,T)$ distribution.  Under $q_{\rm pre}$, fix a regular conditional distribution $q_{\rm pre}(u|s,t)$.  After producing $T^n$, the physical encoder generates the latent coordinates using private randomness according to
\begin{equation}
p(u^n|s^n,t^n)
=
\prod_{i=1}^n q_{\rm pre}(u_i|s_i,t_i).
\label{eq:gaussian-physical-implementation}
\end{equation}
For every $(s_i,t_i)$ in the support of the ordinary $(S,T)$ distribution, this conditional distribution satisfies the defining relation between $U_i$ and $T_i$ almost surely.  Hence applying the unchanged renderer after~\eqref{eq:gaussian-physical-implementation} produces the same Gaussian conditional distribution, i.e., $X_i$ conditioned on $(S_i=s_i,T_i=t_i)$ follows
$\mathcal N\bigl(\mu(s_i)+t_i,\Sigma_Z\bigr).$

Thus the composition of the lift in~\eqref{eq:gaussian-physical-implementation} with the unchanged physical renderer is exactly the effective renderer on every support-compatible $(s^n,t^n)$.  Consequently, the lifted physical-latent code induces the same output distribution and the same decoding performance as the effective-latent code, proving Lem.~\ref{lem:effective-latent}.

We next establish the separate identification consequence needed for Thm.~\ref{thm:gaussian}: exact output preservation fixes the entire key-averaged block distribution of the effective latent $T^n$.  Fix a message $m$ and a context sequence $s^n$.  Let $\tilde T^n$ denote the key-averaged renderer-visible innovation induced by an arbitrary exact code.  That is, for each $i$,
$X_i=\mu(s_i)+\tilde T_i+Z_i,$
where the $Z_i$ are independent $\mathcal N(0,\Sigma_Z)$ and independent of all upstream variables.  For $t^n=(t_1,\ldots,t_n)$, the noise characteristic function is
\begin{equation}
\phi_{Z^n}(t^n)
=
\prod_{i=1}^n
\exp\!\left(-\frac12 t_i^{\mathsf T}\Sigma_Z t_i\right),
\end{equation}
which is nowhere zero.  Exact output preservation gives the desired conditional product distribution $\prod_{i=1}^n q_0(x_i|s_i)$. Hence
\begin{align}
\phi_{X^n-\mu(s^n)|S^n=s^n,M=m}(t^n)
&=
\phi_{\tilde T^n|S^n=s^n,M=m}(t^n)\phi_{Z^n}(t^n),
\label{eq:gaussian-char-factor}\\
\phi_{\tilde T^n|S^n=s^n,M=m}(t^n)
&=
\prod_{i=1}^n
\exp\!\left(-\frac12 t_i^{\mathsf T}\Sigma_T t_i\right),
\nonumber
\end{align}
where the second line follows by division by the nonzero noise factor.
Therefore
\begin{equation}
\tilde T^n\mid(S^n=s^n,M=m)\sim\prod_{i=1}^n\mathcal N(0,\Sigma_T)
\label{eq:gaussian-exact-T}
\end{equation}
for every $(s^n,m)$.  In particular, exact output preservation fixes the entire key-averaged block distribution of the effective latent $\tilde T^n$ in~\eqref{eq:gaussian-exact-T}, including independence across coordinates.  This identification property is used in the finite-key converse below.

\subsection{Key-Conditioned Domination}

Before applying conditional differential-entropy inequalities, we show a regularity consequence of the finite key.  Fix $(s^n,m)$ and a key value $k$ with probability $p(k)>0$.  Let
$\mu_k$ denote the conditional probability measure of $\tilde T^n$ given $(s^n,m,k)$
and let
$G$ denote its key-averaged conditional probability measure given $(s^n,m)$,
which is the Gaussian measure in~\eqref{eq:gaussian-exact-T}.  Since
$G = \sum_k p(k) \mu_k,$
for every measurable set $\mathcal B$,
$p(k)\mu_k(\mathcal B)\le G(\mathcal B).$
Thus $\mu_k\ll G$ and
\begin{align}
\frac{d\mu_k}{dG}
&\le \frac1{p(k)}
\quad G\text{-a.e.},\nonumber\\
D(\mu_k\|G)
&\le \log\frac1{p(k)}.
\label{eq:gaussian-domination}
\end{align}
In particular, each key-conditioned component inherits finite second moments and finite relative entropy with respect to $G$. Differential entropies below are taken on the active Gaussian support, after zero-variance modes are removed. The Gaussian entropy identity and~\eqref{eq:gaussian-domination} then justify the conditional entropy quantities in the EPI argument.

\subsection{Vector Converse}

We next prove the capacity upper bound in Thm.~\ref{thm:gaussian}.  Whiten the renderer noise and rotate into an eigenbasis of $\Gamma$.  On each active mode, write the transformed renderer output as
\begin{equation}
W_j=E_j+N_j,
\quad
E_j\sim\mathcal N(0,\gamma_j),
\quad
N_j\sim\mathcal N(0,1),
\label{eq:gaussian-parallel-modes}
\end{equation}
with independent modes.  Modes with $\gamma_j=0$ contribute zero rate and may be omitted from the differential-entropy calculation.

For the $n$th code, let $R_n$ and $R_{k,n}$ denote its message and key rates. Let $E^n=(E_1^n,\ldots,E_d^n)$ and $W^n=(W_1^n,\ldots,W_d^n)$ collect the length-$n$ sequences across the renderer-visible modes.  By~\eqref{eq:gaussian-exact-T}, the distribution of $E^n$ after averaging over the key is the ordinary product Gaussian distribution and is independent of $(S^n,M)$.  Therefore
\begin{equation}
I(S^n,M,K;E^n)
=
I(K;E^n|S^n,M)
\le
H(K)
\le
nR_{k,n}.
\label{eq:gaussian-key-budget-block}
\end{equation}
Define the modewise key-information quantities
\begin{equation}
r_{j,n}
\triangleq
\frac1n I(S^n,M,K;E_j^n|E_1^n,\ldots,E_{j-1}^n).
\label{eq:rjn}
\end{equation}
The chain rule and~\eqref{eq:gaussian-key-budget-block} give
\begin{equation}
\sum_j r_{j,n}\le R_{k,n}.
\label{eq:rjn-budget}
\end{equation}
Since the unconditional $E$ modes are independent Gaussian,
\begin{equation}
h(E_j^n|S^n,M,K,E_1^n,\ldots,E_{j-1}^n)
=
\frac n2\log(2\pi e\gamma_j)-nr_{j,n}.
\label{eq:gaussian-E-entropy}
\end{equation}
Given $(S^n,M,K,E_1^n,\ldots,E_{j-1}^n)$, the noises in the previous output modes are independent of $W_j^n$, so conditional data processing gives
\begin{equation}
h(W_j^n|S^n,M,K,W_1^n,\ldots,W_{j-1}^n)
\ge
h(W_j^n|S^n,M,K,E_1^n,\ldots,E_{j-1}^n).
\label{eq:gaussian-conditioning-step}
\end{equation}
Since $N_j^n$ is independent of $(E_j^n,S^n,M,K,E_1^n,\ldots,E_{j-1}^n)$, the conditional entropy-power inequality~\cite{ElGamalKim2011}, together with~\eqref{eq:gaussian-E-entropy}, gives
\begin{equation}
h(W_j^n|S^n,M,K,E_1^n,\ldots,E_{j-1}^n)
\ge
\frac n2
\log\!\left(
2\pi e[1+\gamma_j2^{-2r_{j,n}}]
\right).
\label{eq:gaussian-conditional-epi}
\end{equation}
The domination bound~\eqref{eq:gaussian-domination} ensures that the conditional entropy quantities used here are well-defined. Since exact output preservation fixes $h(W_j^n)=\frac n2\log(2\pi e(1+\gamma_j))$, we obtain
\begin{align}
I(S^n,M,K;W^n)
&=\sum_j
\bigl[h(W_j^n)-h(W_j^n|S^n,M,K,W_1^n,\ldots,W_{j-1}^n)\bigr]\nonumber\\
&\le
n\sum_j
\frac12\log
\frac{1+\gamma_j}
{1+\gamma_j2^{-2r_{j,n}}}.
\label{eq:gaussian-converse-sum}
\end{align}

For the converse, we additionally reveal $S^n$ to the detector.  Fano's inequality and the invertible whitening/rotation of the centered output give
\begin{equation}
nR_n
\le
I(M;X^n|S^n,K)+n\epsilon_n
\le
I(S^n,M,K;W^n)+n\epsilon_n,
\quad
\epsilon_n\to0.
\end{equation}
The allocation vectors $(r_{1,n},\ldots,r_{d,n})$ lie in a compact simplex by~\eqref{eq:rjn-budget}; passing to a convergent subsequence and using continuity yields a limiting allocation $r_j\ge0$ with $\sum_jr_j\le R_k$.  Combining with~\eqref{eq:gaussian-converse-sum} gives the right side of~\eqref{eq:gaussian-capacity}.  Since revealing the context $S^n$ can only help the detector, this upper bound applies to both $C_{\ctx}(R_k)$ and $C(R_k)$.  Because the objective is nondecreasing in every $r_j$, any unused key budget may be assigned without decreasing the value, so the maximization may be written with $\sum_jr_j=R_k$ as in Thm.~\ref{thm:gaussian}.

\subsection{Achievability and Key Allocation}

We finally show that the base-model detector attains the converse.  If $\gamma_j=0$ for every $j$, the capacity is zero and the claim is immediate.  Otherwise, in the whitened and rotated coordinates, write
$\bar X = \bar\mu(S)+\bar T+N,$
where the active coordinates satisfy $\bar T_j\sim\mathcal N(0,\gamma_j)$ and $N_j\sim\mathcal N(0,1)$ independently, with both independent of $S$.  This is a vector Gaussian dirty-paper channel with state $\bar\mu(S)$, input $\bar T$, and unit noise covariance.  Let
\begin{equation}
D_\alpha
\triangleq
\operatorname{diag}(\alpha_j),
\qquad
\alpha_j
=
\frac{\gamma_j}{1+\gamma_j},
\end{equation}
and define the standard vector-DPC auxiliary
$
V_0
=
\bar T+D_\alpha\bar\mu(S).
$
The vector Gaussian dirty-paper result gives the Costa cancellation for the state $\bar\mu(S)$~\cite{Costa1983DirtyPaper,YuCioffi2004SumCapacity}.  Since the transformed channel depends on $S$ only through $\bar\mu(S)$, this gives
$
I(V_0;S|\bar X)=0.
$

We impose the finite-key budget by independently degrading this vector-DPC auxiliary.  Choose a nonnegative allocation $(r_1,\ldots,r_d)$ that attains the maximization in~\eqref{eq:gaussian-capacity}; all key can be assigned among modes with $\gamma_j>0$, so $\sum_jr_j=R_k$.  For an active mode with $r_j>0$, let $\nu_j\triangleq\gamma_j/(2^{2r_j}-1)$ and choose independent $N_{V,j}\sim\mathcal N(0,\nu_j)$.  Define
\begin{equation}
V_j
=
V_{0,j}+N_{V,j}
=
\bar T_j
+\alpha_j\bar\mu_j(S)
+N_{V,j}.
\label{eq:gaussian-Costa-V}
\end{equation}
Modes with $r_j=0$ or $\gamma_j=0$ are omitted from the coding auxiliary and generated privately according to their ordinary Gaussian distribution.

Conditioned on $S$, the deterministic shifts disappear, and~\eqref{eq:gaussian-Costa-V} gives
$
I(V;T|S)
=
\sum_j r_j
=
R_k.
$
Because $V$ is obtained from $V_0$ through independent Gaussian degradation, the Costa conditional independence is preserved.  Since the whitening/rotation is invertible, equivalently
$
I(V;S|X)=0.
$
Consequently,
\begin{align}
I(V;X)-I(V;S)
&=I(V;X|S)\nonumber\\
&=\sum_j
\frac12\log
\frac{1+\gamma_j}
{1+\gamma_j2^{-2r_j}}.
\label{eq:gaussian-ach-rate}
\end{align}

It remains to meet the finite-block exactness criterion. The exactification step follows the same conditional maximal-coupling argument as App.~\ref{app:finite-ach}, now applied to the effective latent $T^n$, so we only summarize the Gaussian adaptation.  For this memoryless nonsingular Gaussian law, the general-source/channel soft-covering theorem applies: the relevant information density has finite moments, so the i.i.d. law of large numbers gives the usual mutual-information threshold, and the likelihood-encoder formulation has no finite-alphabet restriction~\cite{Cuff2013Synthesis,SongCuffPoor2016Likelihood}. Apply the likelihood-encoding/soft-covering construction to the effective latent $T$ with the auxiliary in~\eqref{eq:gaussian-Costa-V}.  For message $m$, let $\delta_n(m)$ denote the preliminary synthesis error in total variation, and define the message-average error
$
\bar\delta_n
\triangleq
\frac{1}{|\mathcal M_n|}\sum_{m\in\mathcal M_n}\delta_n(m).
$
The random-coding analysis gives an ensemble for which the expected $\bar\delta_n$ and the expected average decoding error both tend to zero.  Hence a deterministic codebook can be selected for which both quantities vanish.  Retain only messages satisfying
$\delta_n(m)\le\sqrt{\bar\delta_n}.$
By Markov's inequality, the removed fraction is at most $\sqrt{\bar\delta_n}\to0$, so the asymptotic message rate is unchanged and the retained messages satisfy $\max_m\delta_n(m)\to0$.  Since the retained fraction tends to one, their average decoding error also remains vanishing.  Apply the conditional maximal-coupling correction from App.~\ref{app:finite-ach}, now with $T^n$ as the exactified latent.  The correction forces the product Gaussian distribution~\eqref{eq:gaussian-exact-T} exactly for every retained message, and its mismatch probability changes the average decoding error by at most $\max_m\delta_n(m)\to0$.  Finally, generate any inactive modes privately and implement the resulting $T^n$ construction through the physical latent by~\eqref{eq:gaussian-physical-implementation}.  Hence the base-model detector achieves~\eqref{eq:gaussian-ach-rate}, matching the converse.

For completeness, optimizing the allocation in~\eqref{eq:gaussian-capacity} is a standard KKT calculation for a separable concave program~\cite{BoydVandenberghe2004}.  Differentiating the Lagrangian over the active modes gives
$2^{2r_j} = \frac{\gamma_j}{\lambda}$
for a common threshold $\lambda$, which is equivalent to~\eqref{eq:key-allocation}.

\section{Proof of Theorem~\ref{thm:compound}}
\label{app:compound}

\draftEightBegin
In this appendix, we prove the one-pass compound Gaussian capacity in Thm.~\ref{thm:compound} under the random hidden-context information pattern of the base model. For $D\ge D_{\rm res}$, independent resampling is admissible and gives zero capacity, so it remains to consider $D_{\min}\le D<D_{\rm res}$.
\draftEightEnd

\draftEightBegin
Assume $R_k>0$; the case $R_k=0$ follows by continuity and gives zero rate.  Fix a target rate strictly below the right side of~\eqref{eq:compound-capacity}.  By continuity, choose $r$ with $0<r<R_k$
such that the target rate is strictly below
\begin{equation}
C_D(r)
\triangleq
\frac12\log
\frac{\sigma_{\bar X}^2}
{\sigma_{\bar X}^2-
\beta(D)^2\sigma_U^2(1-2^{-2r})}.
\label{eq:CD}
\end{equation}
Let $N_V$ be independent Gaussian noise with variance
$\sigma_U^2/(2^{2r}-1)$ and define the attack-independent Costa-type auxiliary
\begin{equation}
V
=
U+\alpha_D aS+N_V,
\qquad
\alpha_D
\triangleq
\beta(D)\frac{\sigma_U^2}{\sigma_{\bar X}^2}.
\label{eq:compound-Costa-V}
\end{equation}
\draftEightEnd
Then \draftEightBegin
$I(V;U|S)=r
$ and $
\operatorname{Var}(V|S)
\draftEightEnd
=
\draftEightBegin
\frac{\sigma_U^2}{1-2^{-2r}}.
$
\draftEightBegin
The coefficient $\alpha_D$ depends only on the prescribed attack budget, not on the realized attack.

For any admissible attack, exact output preservation~\eqref{eq:gaussian-attack-exact} gives
$\bar Y\triangleq Y-aS\sim\mathcal N(0,\sigma_{\bar X}^2)$ and $\bar Y\perp S$.  Since the context shift cancels in squared error, the distortion constraint gives
\draftEightEnd
\begin{equation}
\Cov(\bar X,\bar Y)
\ge
\sigma_{\bar X}^2-\frac D2
=
\draftEightBegin
\beta(D)\sigma_{\bar X}^2.
\draftEightEnd
\label{eq:compound-cov-XY}
\end{equation}
\draftEightBegin
The attack obeys the conditional Markov relation
$U-(\bar X,S)-\bar Y$, while the original jointly Gaussian pair satisfies
$\EE[U|\bar X,S]
=
\frac{\sigma_U^2}{\sigma_{\bar X}^2}\bar X.
$
\draftEightBegin
Therefore~\eqref{eq:compound-cov-XY} implies
\draftEightEnd
$
\Cov(U,\bar Y)
\ge
\beta(D)\sigma_U^2.
$
\draftEightBegin
Consequently, using $\bar Y\perp S$,
\begin{align}
\operatorname{Var}(V)
&=
\frac{\sigma_U^2}{1-2^{-2r}}
+\alpha_D^2a^2\sigma_S^2,
\label{eq:compound-Vvar}\\
\operatorname{Var}(Y)
&=
\sigma_{\bar X}^2+a^2\sigma_S^2,
\label{eq:compound-Yvar}\\
\Cov(V,Y)
&\ge
\beta(D)\sigma_U^2+\alpha_Da^2\sigma_S^2
=
\alpha_D\bigl(\sigma_{\bar X}^2+a^2\sigma_S^2\bigr).
\draftEightEnd
\label{eq:compound-cov-VY}
\draftEightBegin
\end{align}
Thus both marginals of $(V,Y)$ are fixed Gaussian distributions independent of the attack, while their joint distribution need not be Gaussian.  For fixed Gaussian marginals and fixed covariance, the jointly Gaussian coupling maximizes joint entropy and therefore minimizes mutual information.  Since the covariance floor in~\eqref{eq:compound-cov-VY} is nonnegative over the present distortion range, every admissible attack satisfies
\begin{align}
I(V;Y)-I(V;S)
&\ge
\draftEightEnd
\frac12\log
\frac{\sigma_{\bar X}^2}
{\sigma_{\bar X}^2-
\draftEightBegin
\beta(D)^2\sigma_U^2(1-2^{-2r})}
\nonumber\\
&=C_D(r).
\label{eq:compound-uniform-GP}
\end{align}
The cancellation of $a^2\sigma_S^2$ in~\eqref{eq:compound-uniform-GP} is the Gaussian dirty-paper cancellation adapted to the least favorable covariance permitted by the attack budget.
\draftEightEnd

\draftEightBegin
We next turn the pointwise bound~\eqref{eq:compound-uniform-GP} into one detector that works over the continuum family.  Let $p(v)$ and $p(y)$ denote the fixed Gaussian marginal densities in~\eqref{eq:compound-Vvar}--\eqref{eq:compound-Yvar}, and define the enlarged coupling family
\draftEightEnd
\begin{equation}
\draftEightBegin
\mathcal J_{D,r}
\draftEightEnd
\triangleq
\left\{
\Pi:
\Pi_V(dv)=p(v)\,dv,
\draftEightBegin
\ \Pi_Y(dy)=p(y)\,dy,
\ \int vy\,d\Pi
\ge
\alpha_D(\sigma_{\bar X}^2+a^2\sigma_S^2)
\draftEightEnd
\right\}.
\draftEightBegin
\label{eq:JDr}
\draftEightEnd
\end{equation}
\draftEightBegin
Every physical renderer-realizable attack induces a coupling in $\mathcal J_{D,r}$.  The set of couplings of the two fixed Gaussian marginals is tight and weakly closed, hence weakly compact~\cite{Billingsley1999}.  The covariance half-space is also weakly closed because $VY$ is uniformly integrable over the fixed-marginal coupling family.  Indeed,
\draftEightEnd
\begin{align}
&\sup_{\Pi}
\EE_{\Pi}\!\left[
\draftEightBegin
|VY|\mathbf 1\{|VY|>\tau\}
\draftEightEnd
\right]\nonumber\\
&\quad\le
\draftEightBegin
\sqrt{\operatorname{Var}(Y)}
\draftEightEnd
\left(
\EE[V^2\mathbf 1\{|V|>\sqrt\tau\}]
\right)^{1/2}\nonumber\\
\draftEightBegin
&\qquad+
\draftEightEnd
\sqrt{\operatorname{Var}(V)}
\left(
\draftEightBegin
\EE[Y^2\mathbf 1\{|Y|>\sqrt\tau\}]
\draftEightEnd
\right)^{1/2}
\to0.
\label{eq:compound-uniform-integrability}
\end{align}
\draftEightBegin
Thus $\mathcal J_{D,r}$ is compact.
\draftEightEnd

\draftEightBegin
Let $[x]_b$ denote a nested finite interval quantization of a real variable $x$, with partitions refining to the Borel sigma-field and boundaries of zero probability under the fixed Gaussian marginals. For $\Pi\in\mathcal J_{D,r}$, define $I_b(\Pi)\triangleq I_{\Pi}([V]_b;[Y]_b)$. The same refining quantizer sequence can be chosen independently of the coupling, and
$I_b(\Pi)\uparrow I_{\Pi}(V;Y)$
for every $\Pi\in\mathcal J_{D,r}$~\cite[Lemma~5.5.5]{Gray2011Entropy}.  For each fixed $b$, $I_b(\Pi)$ is continuous in $\Pi$ under weak convergence.  Since the target rate is strictly below $C_D(r)$, choose $\epsilon>0$ such that
\begin{equation}
R+I(V;S)+3\epsilon<C_D(r)+I(V;S)
\le
\inf_{\Pi\in\mathcal J_{D,r}}I_{\Pi}(V;Y).
\label{eq:compound-strict-margin}
\end{equation}
The increasing open sets
\draftEightEnd
\begin{equation}
\mathcal O_b
\triangleq
\draftEightBegin
\left\{
\Pi\in\mathcal J_{D,r}:
I_b(\Pi)>R+I(V;S)+2\epsilon
\right\}
\draftEightEnd
\end{equation}
\draftEightBegin
cover the compact set $\mathcal J_{D,r}$. Hence there exists one finite $b_*$ such that
\draftEightEnd
\begin{equation}
\draftEightBegin
\inf_{\Pi\in\mathcal J_{D,r}}
I_{\Pi}\bigl([V]_{b_*};[Y]_{b_*}\bigr)
\draftEightEnd
>
\draftEightBegin
R+I(V;S)+2\epsilon.
\draftEightEnd
\label{eq:uniform-quantized-margin}
\end{equation}
\draftEightBegin
This finite detector interface is attack independent and uses the actual attacked observation $Y$, not the unobservable residual $Y-aS$.
\draftEightEnd

\draftEightBegin
Because $[V]_{b_*}$ is a function of $V$, data processing gives
\begin{align}
I([V]_{b_*};S)
&\le I(V;S),\label{eq:quantized-state-budget}\\
I([V]_{b_*};U|S)
&\le I(V;U|S)=r<R_k.
\label{eq:quantized-key-budget}
\end{align}
Choose $0<\eta<\epsilon$ and set the total multicoding rate
$
\hat R=I([V]_{b_*};S)+\eta.
$
Then the finite detector interface and the key margin imply the three raw coding inequalities
\begin{align}
\hat R&>I([V]_{b_*};S),\label{eq:compound-cover-state}\\
R+\hat R&<
\inf_{\Pi\in\mathcal J_{D,r}}
I_{\Pi}([V]_{b_*};[Y]_{b_*}),
\label{eq:compound-pack-finite}\\
R_k+\hat R&>I([V]_{b_*};S,U).
\label{eq:compound-cover-latent}
\end{align}
The first two inequalities have the usual state-covering and compound-packing structure of a common-auxiliary Gel'fand--Pinsker construction~\cite{PiantanidaShamai2010CompoundState}; the third is the additional latent-coordination requirement imposed by exact preservation. The project-wide excess variable $t=\hat R-I(V;S)$ is a derived quantity and is not the raw codebook rate in this argument.

Generate one i.i.d. $p([v]_{b_*})$ codebook indexed by $(m,\ell,k)$ with
$
\ell\in[1:2^{n\hat R}].
$
Given $(S^n,m,k)$, use the likelihood encoder associated with the channel $[V]_{b_*}\to S$ to select $\ell$, and generate a preliminary latent sequence through the memoryless kernel $p(u|s,[v]_{b_*})$. For a fixed key index, the detector quantizes the attacked public sample coordinatewise to $[Y]_{b_*}^n$ and decodes $(m,\ell)$ over the finite compound family induced by~\eqref{eq:uniform-quantized-margin}.

The universal decoder must work for the same i.i.d. $p([v]_{b_*})$ codebook ensemble used by the likelihood encoder and synthesis layer. For this finite-alphabet memoryless family, apply the universal decoder of~\cite{Merhav2013Universal} relative to the collection of matched maximum-likelihood metrics. Its ensemble-average error under the given random-coding distribution is within a subexponential factor of the best decoder in that class. The strict uniform information margin in~\eqref{eq:uniform-quantized-margin}, together with compactness of the induced channel family, gives a uniformly positive matched-decoding random-coding exponent; hence this same i.i.d. ensemble has vanishing expected worst-attack average decoding error under the corresponding ideal memoryless law. The likelihood-encoder transfer from the ideal state law to the operational observed state is controlled by~\eqref{eq:compound-cover-state}; every renderer and attack is downstream, so total variation contracts uniformly over $A\in\Areg(D)$.

The synthesis layer is attack independent. For each fixed $(m,k)$, the $\ell$-subcodebook has total rate $\hat R>I([V]_{b_*};S)$, so ordinary general-alphabet soft covering yields vanishing expected state-covering total variation. For each fixed message $m$, the aggregate $(k,\ell)$-indexed subcodebook has rate $R_k+\hat R>I([V]_{b_*};S,U)$, so the same general-alphabet likelihood-encoder/soft-covering machinery gives vanishing expected synthesis total variation for the joint $(S,U)$ law~\cite{Cuff2013Synthesis,SongCuffPoor2016Likelihood}. Averaging over the relevant indices, together with the worst-attack decoding bound above, allows one deterministic codebook to be selected for which all three average criteria vanish. Expurgate the $o(1)$ fraction of messages with atypically large preliminary synthesis error. This preserves the asymptotic message rate and worst-attack average reliability while making the preliminary synthesis error uniformly vanishing over the retained messages.
\draftEightEnd
For every retained message, apply the conditional maximal-coupling exactification from App.~\ref{app:finite-ach} before the physical renderer.  The correction forces
\draftEightEnd
\begin{equation}
\draftEightBegin
p(u^n|s^n,m)
=
\prod_{i=1}^n q_{\rm lat}(u_i|s_i)
\label{eq:compound-exact-latent}
\draftEightEnd
\end{equation}
\draftEightBegin
exactly.  Passing this latent through the fixed renderer gives the exact desired conditional output law before attack, and every $A\in\Areg(D)$ preserves that law by definition.  If the exactification mismatch probability for message $m$ is $\delta_n(m)$, total variation contracts under the renderer and every downstream attack, so the reliability change is at most $\delta_n(m)$ uniformly over the entire attack family.  This proves achievability of every rate strictly below~\eqref{eq:compound-capacity}.

For the converse, the canonical Gaussian regeneration~\eqref{eq:canonical-attack} belongs to the admissible renderer-realizable family.  Expanding its residual recursion gives the effective renderer
\draftEightEnd
\begin{equation}
\draftEightBegin
Y
=
aS+\beta(D)U+Z_{\rm eff},
\label{eq:compound-canonical-effective}
\draftEightEnd
\end{equation}
\draftEightBegin
where $Z_{\rm eff}\perp(S,U)$ is Gaussian with variance $\sigma_{Z,{\rm eff}}^2=\sigma_{\bar X}^2-\beta(D)^2\sigma_U^2$. This is a scalar instance of the jointly Gaussian renderer in Thm.~\ref{thm:gaussian}. That theorem already allows the context to remain hidden from the detector and shows that revealing it causes no capacity increase. Evaluating~\eqref{eq:gaussian-capacity} at the effective SNR $\gamma_{\rm eff}(D)=\beta(D)^2\sigma_U^2/[\sigma_{\bar X}^2-\beta(D)^2\sigma_U^2]$ gives exactly the right side of~\eqref{eq:compound-capacity}.  Since every compound code must work against this particular admissible attack, this fixed-attack capacity upper-bounds the compound capacity.  Combined with the uniform achievability above, it proves Thm.~\ref{thm:compound}.
\draftEightEnd

\ifdefined\SUPPRESSACKNOWLEDGMENT
\else
\section*{Acknowledgment}
The problem formulation, theorem statements, and proof strategy for each result were developed by the authors. The random-projection construction in Proposition~\ref{prop:strict-converse-example} was proposed by OpenAI ChatGPT in response to the authors' question of whether the outer bound in Theorem~\ref{thm:finite} can be strictly loose; the authors verified the construction and developed its interpretation. The detailed proofs were developed interactively with ChatGPT: for each result, the authors specified the proof structure and key steps, ChatGPT produced candidate arguments, and these were iteratively corrected and tightened under the authors' direction. For steps following standard arguments, including soft covering, Fourier--Motzkin elimination, and the Csisz\'ar sum identity, ChatGPT assisted in drafting routine details. ChatGPT also assisted with drafting and revising expository text, literature screening, and writing numerical-check code. Anthropic Claude was used to independently re-derive the proofs, reproduce the numerical results from the stated formulas, spot-check references, and provide editorial feedback. The authors reviewed every proof line by line, checked the cited technical sources against the original references, and take full responsibility for the final content.
\fi

\bibliographystyle{IEEEtran}
\bibliography{references/citations,references/draft8_additions}

\end{document}